%% file: Kittel-TimingOfTransients.tex
\documentclass[10pt]{iopart}  

\usepackage{iopams} 

\include{additional}

\begin{document}

\title[Timing of Transients]{Timing of Transients: Quantifying Reaching Times and Transient Behavior in Complex Systems}

\author[pik]{Tim Kittel$^{1,2}$, Jobst Heitzig$^1$,  Kevin Webster$^1$, Jürgen Kurths$^{1,2,3}$}

\address{%
	$^1$ Potsdam Institute for Climate Impact Research, 
	Telegrafenberg A31 - (PO) Box 60 12 03,
	14412  Potsdam, Germany
}%
\address{%
	$^2$ Institut für Physik, Humboldt-Universität zu Berlin, 
	Newtonstraße 15,
	12489 Berlin, Germany
}%
\address{%
	$^3$ Institute for Complex Systems and Mathematical Biology, University of Aberdeen, 
	Aberdeen AB24 3UE, United Kingdom
}
\ead{Tim.Kittel@pik-potsdam.de}

\vspace{8pt}

\begin{indented}
	\item[\textbf{Submitted on:}]\ soon
\end{indented}

\begin{abstract}
\input{abstract}
\vspace{10pt}
\begin{description}
	\item[Submitted to:]\ \NJP
	\vspace{8pt}
	\item[PACS:]\ 05.45.-a, 02.30.Hq
	\vspace{8pt}
	\item[Keywords:]\ Complex Systems, Nonlinear~Dynamics, Long~Transients, Stability~against~Shocks, Ordinary~Differential~Equations, Early-Warning~Signals
\end{description}
\end{abstract}




\maketitle

 
\ioptwocol

\input{Kittel-TimingOfTransients-maintext.tex}

\ack

This paper was developed within the scope of the IRTG 1740/TRP 2011/50151-0, funded by the DFG/FAPESP.
This work was conducted in the framework of PIK’s flagship project on coevolutionary pathways (\textsc{copan}).
The authors thank CoNDyNet (FKZ 03SF0472A) for their cooperation.
The authors gratefully acknowledge the European Regional Development Fund (ERDF), the German Federal Ministry of Education and Research and the Land Brandenburg for supporting this project by providing resources on the high performance computer system at the Potsdam Institute for Climate Impact Research.
The authors thank the developers of the used software: Python\cite{python}, Numerical Python\cite{numpy} and Scientific Python\cite{scipy}.

The authors thank
Sabine Auer,
Karsten Bölts,
Catrin Ciemer,
Jonathan Donges,
Jasper Franke,
Frank Hellmann,
Jakob Kolb,
Chiranjit Mitra,
Finn Müller-Hansen,
Jan Nitzbon,
Reik Donner,
Stefan Ruschel,
Tiago Pereira da Silva,
Francisco A. Rodrigues,
Paul Schultz,
and Lyubov Tupikina
for helpful discussions and comments.

\appendix

\input{Kittel-TimingOfTransients-appendix1.tex}

\input{Kittel-TimingOfTransients-appendix2.tex}

\section*{References}
\bibliography{mybib}

\end{document}

%% file: additional.tex
\usepackage[english]{babel}
\usepackage[utf8]{inputenc}

\usepackage{xparse}

\usepackage{cite}

\newcommand\figurespath{figures/}
\usepackage{graphicx}
\usepackage[svgpath = {\figurespath}]{svg}
\graphicspath{{\figurespath}}
\usepackage{psfrag}
\usepackage[percent]{overpic}
\usepackage{array} 
\newlength{\vpixdist}  

\expandafter\let\csname equation*\endcsname\relax
\expandafter\let\csname endequation*\endcsname\relax
\usepackage{amsmath,
	amssymb,
	 amsthm,
	 bbold,
	 bm} 
\usepackage{mathtools}
\renewcommand{\arraystretch}{1.2}  

\usepackage{cleveref}
\crefname{figure}{fig\@.}{figs\@.} 
\Crefname{figure}{Fig\@.}{Figs\@.}
\crefname{section}{sec\@.}{secs\@.} 
\Crefname{section}{Sec\@.}{Secs\@.}
\crefname{equation}{eq\@.}{eqs\@.} 
\Crefname{equation}{Eq\@.}{Eqs\@.}

\usepackage{dcolumn}

\NewDocumentCommand\rrt{g}{\ensuremath{T_{RR\IfNoValueTF{#1}{}{#1}}}}
\newcommand{\audic}{\ifmmode{}D\else{}\textsc{a}u\textsc{d}i\textsc{c}\fi}

\newcommand\inv{^{-1}}
\newcommand\p{^\prime}
\newcommand\R{\mathbb{R}}

\newcommand{\attr}{\mathcal{A}}
\newcommand{\basin}{\mathcal{B}}
\newcommand{\basinA}{\basin_\attr}

\newcommand{\id}{\mathbb{1}}
\renewcommand{\Re}{\text{Re}}

\newlength{\elementwidth}
\newlength{\elementwidthtwo}

\newtheorem{theorem}{Theorem}[section] 
\newtheorem{definition}[theorem]{Definition} 
\newtheorem{proposition}[theorem]{Proposition} 
 
\newtheorem{lemma}[theorem]{Lemma}

%% file: abstract.tex
\noindent
When quantifying the time spent in the transient of a complex dynamical system, the fundamental problem is that for a large class of systems the actual time for reaching an attractor is infinite. Common methods for dealing with this problem usually introduce three additional problems: non-invariance, physical interpretation, and discontinuities, calling for carefully designed methods for quantifying transients.
	
In this article, we discuss how the aforementioned problems emerge and propose two novel metrics, \emph{Regularized Reaching Time} (\rrt{}) and \emph{Area under Distance Curve} (\audic{}), to solve them, capturing two complementary aspects of the transient dynamics.
	
\rrt{} quantifies the additional time (positive or negative) that a trajectory starting at a chosen initial condition needs to reach the attractor after a reference trajectory has already arrived there. A positive or negative value means that it arrives by this much earlier or later than the reference.	
Because \rrt{} is an analysis of return times after shocks, it is a systematic approach to the concept of critical slowing down \cite{lenton2011early}; hence it is naturally an early-warning signal \cite{scheffer2009early} for bifurcations when central statistics over distributions of initial conditions are used.

	
\audic{} is the distance of the trajectory to the attractor integrated over time. Complementary to \rrt{}, it measures which trajectories are \emph{reluctant}, i.e.\ stay away from the attractor for long, or \emph{eager} to approach it right away.
	
Four paradigmatic examples have been chosen in order to display the different features of these novel metrics and their relations: a linear system, a global carbon cycle model \cite{anderies2013topology}, a generator in a power grid \cite{yuan2003solution} and the chaotic Rössler attractor \cite{rossler1976equation}. While the linear system can be solved analytically, we demonstrate our efficient algorithms for the three non-linear examples using the fact that the metrics are Lyapunov functions \cite{giesl2007construction}. New features in these models can be uncovered, including the surprising regularity of the Rössler system's basin of attraction even in the regime of a chaotic attractor. Additionally, we demonstrate the \emph{critical slowing down} interpretation by presenting the metrics' sensitivity to prebifurcational change and thus how they act as \emph{early-warning signals}.


%% file: Kittel-TimingOfTransients-maintext.tex
\section{Introduction}
	
%
%
%
	
	In complex dynamical systems, the importance of a trajectory's transient, i.e.\ the part of the trajectory ``away'' from the attractor, plays an important role in physics research, e.g. for lasers \cite{fiutak1980transient,tang1975transient}, the dynamical Ising model \cite{barkema1994transient} and other parts of statistical physics \cite{castellano2009statistical,chowdhury2000statistical,krapivsky2010kinetic} as well as in various other fields, including ecology \cite{van2007long,hastings2004transients}, biology \cite{schaffer1993transient}, economics \cite{fisher1989disequilibrium}, medicine \cite{fischer2005epileptic} and climate change \cite{lenton2011early,anderies2013topology} with specific focus on \emph{long transients} in \cite{van2007long,anderies2013topology,van2016constrained}. Hastings~\cite{hastings2004transients} stresses the importance of different time scales and points out how the transient dynamics can be very different and much more interesting than the asymptotic behavior. In addition, he points out how saddles play a central role by inducing long transients.
	
	In this article, we devise novel metrics that measure how long it takes to reach the system's attractor to foster the study of transients.
		
	Even though common methods for that exist, they are confronted with four essential problems.
	(I) \emph{divergences}: the attractor is not reached in finite time for a large class of physically relevant systems; 
	(II) \emph{physical interpretation}: when using $\epsilon$-neighborhoods the results depend strongly on the choice of $\epsilon$ and similarly for other parametrized methods;
	(III) \emph{discontinuities}: small changes in the parameter often have a large (noncontinuous) effect on the measured time;
	and (IV) \emph{non-invariance}: the results depend on the choice of variables. Problem (IV) is particularly important, as a result should be a property of the dynamical system and thus independent of the choice of variables, i.e.\ invariant (or correctly transforming) under smooth transformations of the state space (cf. ``smoothly equivalent'' in \cite{kuznetsov2013elements}).
	
	As these problems are fundamental and have not yet been solved, we do not aim to just improve a quantitative measure but actually provide new solutions to general properties and consistency requirements.
	
	The two novel metrics proposed in this article, the \emph{Regularized Reaching Time} (\rrt) and the \emph{Area under Distance Curve} (\audic), have been developed in order to treat the aforementioned problems. The first is defined by the difference of the reaching times with respect to a reference trajectory and thus actually measures a time. The idea is that even though the actual reaching times diverge (problem (I)), the difference converges and we can compare which trajectories reach the attractor \emph{earlier} or \emph{later}. The second one, \audic, is the distance to the attractor integrated along the trajectory. This means that it takes a different point of view and measures which trajectories are \emph{reluctant}, i.e.\ stay away from the attractor for long, or \emph{eager}, i.e.\ approach it right away.
	
	Both metrics are shown to be Lyapunov functions and this property of \audic{} is used for the computation. In the outlook, we even suggest this property to be the basis of an improved definition of \rrt{}.
	
	When applying these metrics to the global carbon cycle model \cite{anderies2013topology} and the chaotic Rössler oscillator \cite{rossler1976equation}, their potential as \emph{early-warning signals} \cite{scheffer2009early,lenton2011early} becomes apparent. Statistics of their distributions in state space represent the system's \emph{critical slowing down} (CSD) \cite{scheffer2009early,scheffer2012anticipating,lenton2011early} after a shock, i.e.\ an instantaneous and non-infinitesimal perturbation, uncovering prebifurcational changes in the transient behavior. In contrast, CSD is usually done with (local) noise only. The usage of shocks has been developed in the context of Basin Stability \cite{menck2013basin,menck2014dead} and its extensions \cite{klinshov2015stability,hellmann2015survivability,van2016constrained,mitra2015integrative}.
	
	With this new approach, we have been able to uncover new features of the systems: the basin of attraction in the chaotic Rössler system is unexpectedly regular and the basin separation in the carbon cycle is due to the strong stable manifold acting as a separatrix induced by a saddle, demonstrating the idea how saddles can lead to long transients. Additionally, we show how the metrics work well as \emph{early-warning signals} by detecting the prebifurcational changes.
	
	The remainder of this article is structured as follows. After stating the fundamental problems of reaching time definitions in \Cref{sec:problems-reaching-time}, we present two complementary solutions in \Cref{sec:novel-metrics} and apply them to examples in \Cref{sec:examples} before concluding with a discussion and some outlook in \Cref{sec:discussion}. The Appendix contains more technical comments, calculations and additional information. Moreover, \ref{app2-sec:definitions-and-theorems} contains mathematical definitions and proofs, putting the ideas presented in the main text on solid footing.

\section{The problems of reaching time definitions}
\label{sec:problems-reaching-time}

\begin{figure}
	\includegraphics[width = \columnwidth]{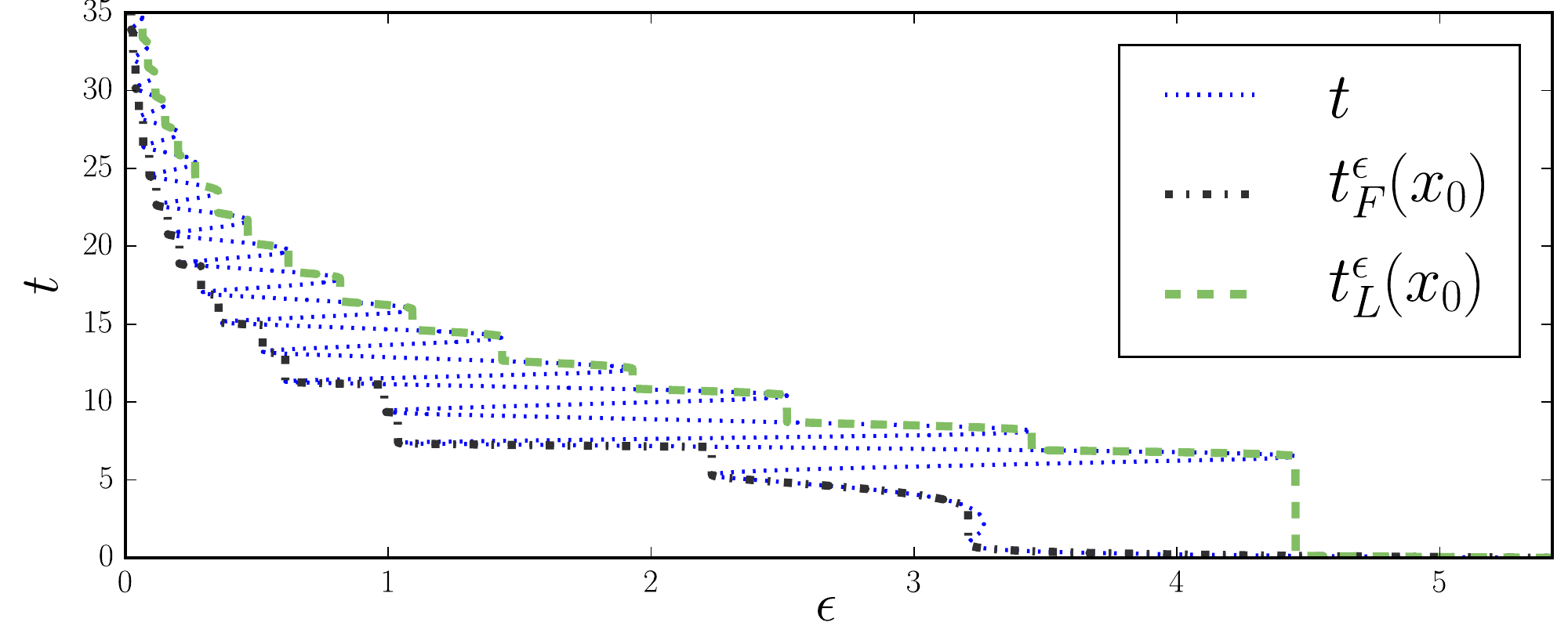}
	\caption{(color online) The distance (x-axis) over time (y-axis) for an example system with a stable, spiraling fixed point and a saddle, chosen to show the occurrence of problems (I)-(IV) for common $\epsilon$-dependent definitions of reaching times, $t^\epsilon_1(x_0)$ and $t^\epsilon_L(x_0)$, as discussed in the text.}
	\label{fig:eps-t}
\end{figure}

%
%
%

	A basic property of a large class of complex systems is that trajectories reach the attractor in infinite time only, even for steady states or limit cycles and generally most systems of ordinary differential equations with smooth right hand side functions. 
		 
	Two common metrics that work around this simply measure when the trajectory starting at $x_0$ enters an $\epsilon$-neighborhood around the attractor the first time ($t^\epsilon_F(x_0)$) or the last time ($t^\epsilon_L(x_0)$), i.e.\ when the neighborhood is not left anymore. In \Cref{fig:eps-t}, the (Euclidean) distance $\epsilon$ (x-axis, dotted light blue line) vs. the time $t$ (y-axis) of a basic example system's trajectory have been plotted (more details and explanations in \ref{app1-sec:divergence-problem}); $t^\epsilon_F(x_0)$ and $t^\epsilon_L(x_0)$ have been added.
	
	Firstly, we observe the divergence of $t$ for $\epsilon \rightarrow 0$, corresponding to the actual time until reaching the attractor (problem (I)).
	
	Because $t$ diverges, $t^\epsilon_F(x_0)$ and $t^\epsilon_L(x_0)$ depend heavily on the choice of $\epsilon$ so that a proper physical interpretation is rather difficult (problem (II)). It is far from being obvious what ``close'' or ``when the transient is over'' means. Thirdly, the strong discontinuities for $t^\epsilon_F(x_0)$ and $t^\epsilon_L(x_0)$ in \Cref{fig:eps-t} make the choice of a proper $\epsilon$ even harder (problem (III)). 
	
	Problem (IV) from the introduction is non-invariance: Using a different set of variables, i.e.\ smoothly transforming the system, gives different values for $t^\epsilon_F(x_0)$ and $t^\epsilon_L(x_0)$, because the Euclidean distance is not invariant. This means the results are not a genuine property of the dynamical system but just of the representation.
	
	Other metrics are based on characteristic times \cite{nolting2011grundkurs} and Lyapunov exponents \cite{cvitanovic2016chaos}. Though common, the former suffer the same problems as the approaches above and are constant for a 1-dimensional linear system which is counter-intuitive when thinking about reaching times. The latter share these problems but are invariant under changes of variables, i.e.\ they are physical in that sense. Note that Lyapunov exponents do not capture the transient at all but are an asymptotic feature of the system only.
	
	An extended discussion of these problems including an exemplary model is given in \ref{app1-sec:divergence-problem}.

\section{Two novel, complementary solutions}
\label{sec:novel-metrics}

	To overcome the above problems, we introduce two metrics: \emph{Area under Distance Curve} (abbreviated as \audic{}, $\audic{}$) and \emph{Regularized Reaching Time} (\rrt), and show that they naturally lead to a transient analysis from separate points of view.
		
	(i) \emph{Area under Distance Curve} ($\audic{}$) comes from the idea that a trajectory stays far away from the attractor during the transient while it is close in the asymptotics. A distance function $d(\cdot, \cdot)$ is needed to have notions of ``far'' and ``close'' and we define \audic{} as 
	\begin{align}
	D(x_0) = \int_{0}^{\infty} \text{d}t\ d(x(t),\, &\attr) \qquad \text{with } x(0) = x_0,    \label{eq:audic-definition}
	\end{align}
	where $\attr$ is the attractor with the basin $\basinA$ and $x(t)$ the trajectory. Hence, we look at the cumulative distance and remove the influence of the asymptotics.
	
	Note the strong difference of $D$ to $t^\epsilon_F(x_0)$ and $t^\epsilon_L(x_0)$. Both of the latter are very sensitive to small changes of the distance function around the attractor and to $\epsilon$; and it is difficult to even find a sensible notion of distance. In the \audic{} measure $D$ instead, choosing a tailor-made distance function $d(\cdot, \cdot)$ allows to adapt the measure to specific research questions, e.g., by letting $d(x,\attr)$ represent some form of costs or damages due to being away from the attractor. This approach solves problem (IV), too, because the distance function is transformed correctly.
		
	Initial conditions with high values of \audic{} are called \emph{reluctant} and low values \emph{eager}. This terminology is used to emphasize that \emph{reluctant states} go through large transients far away from the attractor, while \emph{eager states} approach it ``right away''.
	
	In the case where $\mathcal{A}$ is a hyperbolic fixed point of an ODE  $\dot{x} = f(x)$, we show under certain mild conditions on $d(\cdot, \cdot)$ that \audic{} is a \emph{Lyapunov function} \cite{giesl2007construction} uniquely defined by
	$D(\mathcal{A}) = 0$
	and
	$\frac{d}{dt}D(x) = - d(x,\mathcal{A})$, where 
	$\frac{d}{dt}D(x) =  \left(\nabla D(x)\right)^T \cdot f(x)$
	is the orbital derivative. Thus, the level sets of $D$ foliate the basin of attraction and are forward invariant under the flow, which will be used for the definition of \rrt{}.
	
	Proofs and further properties are given in \ref{app1-sec:1d-systems}, \ref{app1-sec:conv-audic} and \ref{app2-sec:AUDIC}, incl. a convergence discussion of \Cref{eq:audic-definition}.

	(ii) \emph{Regularized Reaching Time} (\rrt{}) is based on time differences between trajectories. It can be interpreted as the additional time (positive or negative) that a trajectory starting at a point of interest $x_0$ needs to reach the attractor after a reference trajectory has already arrived there. A positive or negative value means that the trajectory at hand arrives by this much later or earlier, respectively, at the attractor than the reference trajectory does. Since the actual reaching times are both infinite, \rrt{} is formally defined as the limit for $\epsilon \longrightarrow 0$ of the difference between how long the trajectory at hand and the reference trajectory need to reach the \audic{} level set with value $\epsilon$.

	This idea is put in equations as follows. First we define the time $t^\text{\audic}\left(x_0, \epsilon\right) :=t$ where $\epsilon = D\left(x\left(t\right)\right)$ and $x(t) = x_0$. Note that the forward invariance of \audic{} provides uniqueness of $t^\text{\audic}$. Next, we choose the initial condition $x_{ref}$ of the reference trajectory. And finally, we define the \emph{Regularized Reaching Time} by taking the limit
	\begin{align}
	\rrt (x_0) := & \lim_{\epsilon \rightarrow 0} \left(t^\text{\audic}\left(x_0, \epsilon\right) - t^\text{\audic}\left(x_{ref}, \epsilon\right) \right). \label{eq:def-rrt}
	\end{align}
	While the existence of this limit is far from obvious, we can prove it under mild conditions for an important class of systems: namely hyperbolic fixed points (see \ref{app2-sec:RRT}). Problem (IV), non-invariance, is circumvented because for any smooth invertible coordinate transformation $\Phi$ the equation $\rrt (\Phi(x_0)) = T\p_{RR}(x_0)$ holds (also \ref{app2-sec:RRT}), where $T\p_{RR}$ is the regularized return time \rrt{} in the system with changed variables.
 	
 	Our numerical results indicate that this idea is sensible for more complex attractors, too, particularly the limit cycle discussed in \Cref{sec:swing} and the chaotic Rössler system below.
	 
	\rrt{} represents \emph{the actual time} by how much a trajectory reaches the attractor \emph{later} or \emph{earlier} than the one starting at the reference point. Different choices of $x_{ref}$ result in additive constants such that $\rrt  \left(x_{ref}\right) = 0$, but do not affect the structure of $\rrt$. Thus central statistics, i.e.\ ones invariant under shifts, are sensible; especially the standard deviation proves useful for the examples below.
	
	Note that we used the \audic{} level sets in the definition for \rrt{} as they are parametrized, bounded, forward invariant foliations. Foliations like this are usually hard find or compute but for these particular ones an efficient algorithm was developed (see \ref{app1-sec:algo}) and their computation can be done even for chaotic systems. Also, the usage of \audic{} avoids local geometric measures that can easily induce problem (IV) (non-invariance).
	
	We find also that \rrt{} is a Lyapunov function, this time with constant (negative) orbital derivative $-1$; and in \ref{app2-sec:RRT} we show that for hyperbolic fixed-point attractors, \rrt{} is a time-parametrization of the strong stable foliation $\mathcal{F}_{ss}$ in $\basinA$ (see \Cref{app2-thm:Fss}) with respect to a reference leaf. These properties could be used in order to find an improved definition of \rrt{} (as discussed later). It follows that \rrt{} diverges to $-\infty$ as it approaches $\mathcal{A}$ or its strong stable manifold $W^{ss}(\mathcal{A})$ (if it exists) which is the manifold associated to the smaller Lyapunov exponents (precise definition in \cite{hirsch1977invariant}). This implies that during the phase space estimation $W^{ss}(\mathcal{A})$ becomes visible as e.g. in the carbon cycle example below.
	(See \ref{app1-sec:1d-systems}, \ref{app1-sec:conv-rrt} and \ref{app2-sec:RRT} for precise definitions, proofs and further properties.)
	
	In contrast to the convergence of the distance between trajectories in isochrons \cite{mauroy2013isostables,josic2006isochron} this work focuses on the distance to the attractor, giving rise to this timing of transients.	

\section{Examples}
\label{sec:examples}

	In order to demonstrate the applicability of the metrics, we selected four examples with differing properties. They have been chosen with increasing levels of complexity and to show different properties of the new metrics. Note that 1-dimensional systems can be solved analytically and we discuss an additional example in \ref{app1-sec:1d-systems}.
	
	\bgroup
	\setlength{\elementwidth}{0.28\textwidth}
	\setlength{\elementwidthtwo}{\elementwidth}
	\setlength{\tabcolsep}{0pt}
	\setlength\vpixdist{-4.4mm}
	\newcolumntype{C}[1]{>{\centering\let\newline\\\arraybackslash\hspace{0pt}\vspace{\vpixdist}}m{#1}}
	\renewcommand*{\arraystretch}{0.2}
	\begin{figure*}
		\centering
		\vspace{-5mm}
		\subfloat{\includegraphics[width = 0.49\textwidth]{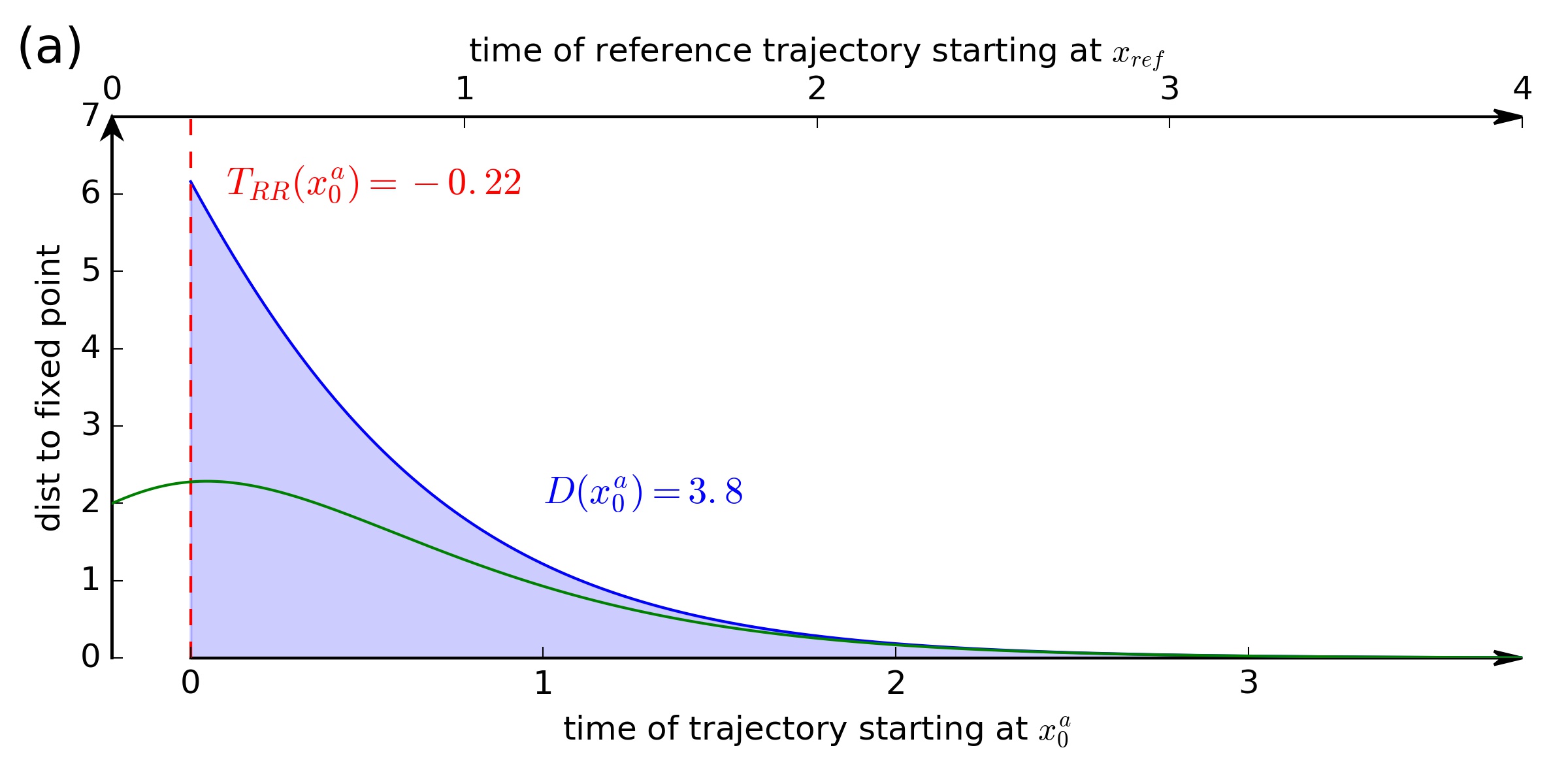}}
		\hfill
		\subfloat{\includegraphics[width = 0.49\textwidth]{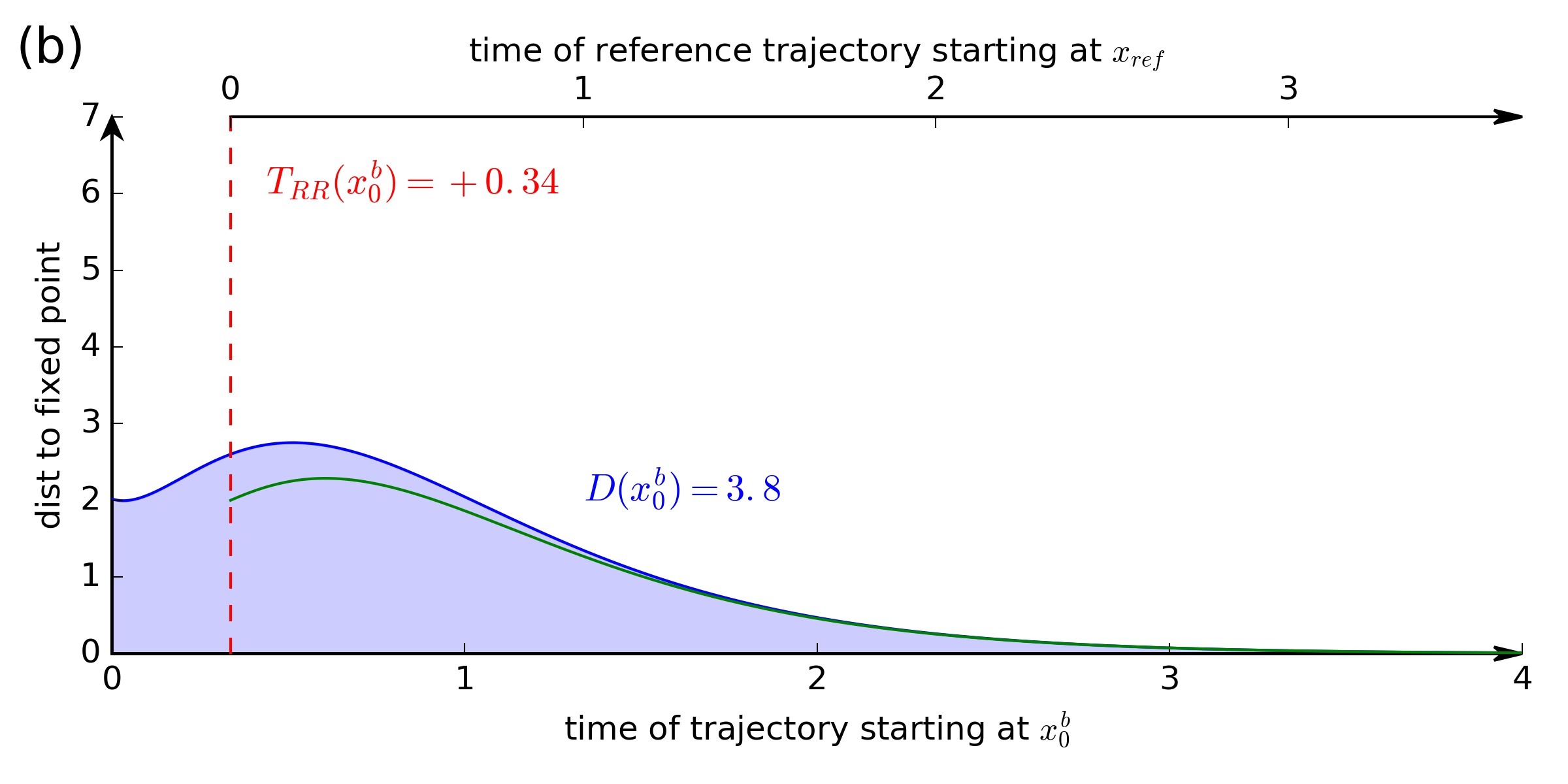}}\\
		\vspace{-5mm}
		\caption{(color online) The figure shows for two exemplary initial conditions (a) $x^a_0 = (0.8, 2.35)^T$ and (b) $x^b_0=(1.4, 0.24)^T$ the distance of the attractor over time (blue curve) in the linear example system of \Cref{sec:linear-system}. The initial conditions have been chosen such that the \audic{} value, which corresponds to the blue-shaded area, is the same for both trajectories, $\audic(x^a_0) = \audic(x^b_0) = 3.8$. But the trajectory starting at $x^a_0$ arrives \emph{earlier} than the reference trajectory (green in (a) and (b)), which in turn is \emph{earlier} than the one from $x^b_0$, meaning $\rrt{}(x^a_0) = -0.22 < \rrt{}(x_{ref}) = 0 < \rrt{}(x^b_0) = +0.34$. In order to show this, the example trajectories (blue) have been shifted in each plot by the value of \rrt{} with respect to the reference trajectory. This demonstrates the intuition behind \rrt{}: it describes by how much one has to shift one trajectory so it matches the asymptotics of the reference trajectory. The proofs in \ref{app2-sec:definitions-and-theorems} provide that this is always possible for a generically chosen reference point $x_{ref}$.}
		\label{fig:linear-example}
		\vspace{-2mm}
		
		\begin{tabular}[width=30cm]{C{0.31\textwidth} C{0.31\textwidth} C{0.27\textwidth}}
			\subfloat{\includegraphics[width = \elementwidth]{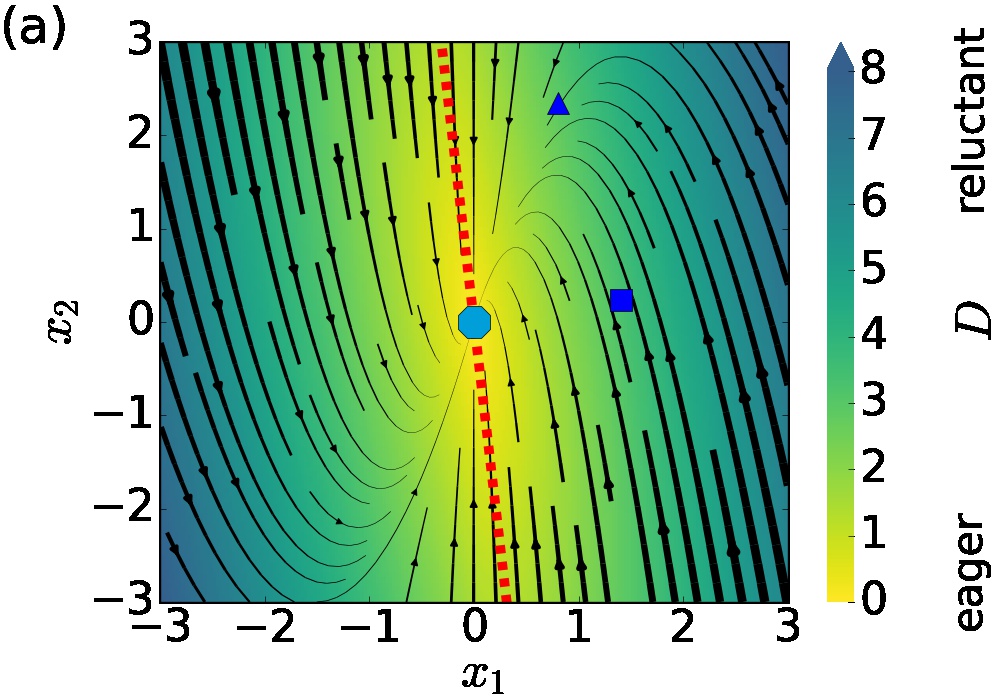}\label{fig:linear-nsp-audic}} & 
			\subfloat{\includegraphics[width = \elementwidthtwo]{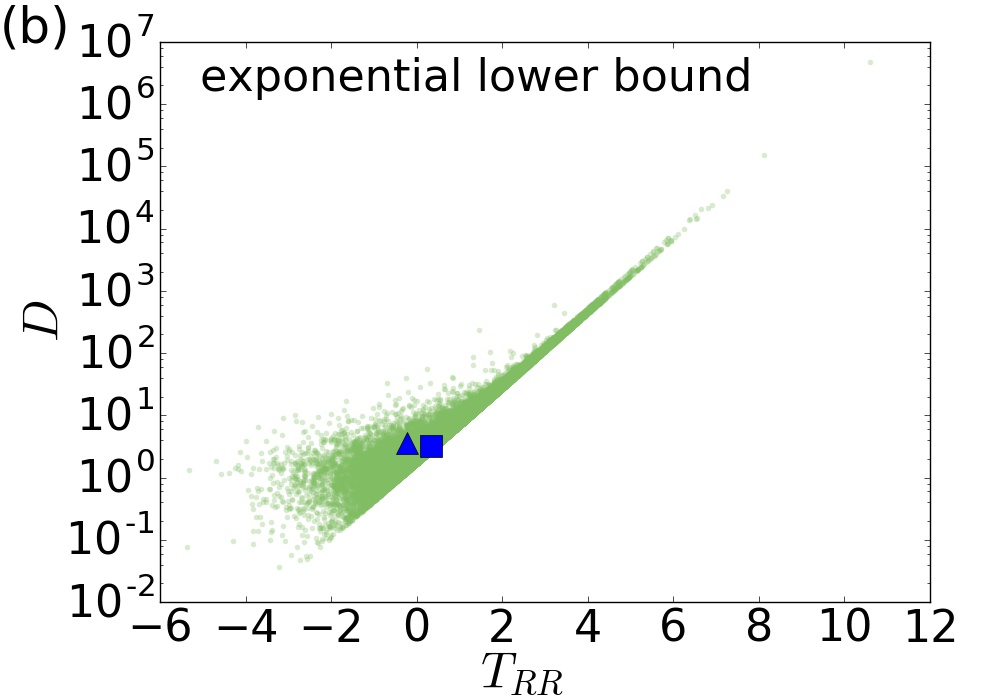}\label{fig:linear-nsp-scatter}} &
			\subfloat{\includegraphics[width = \elementwidth]{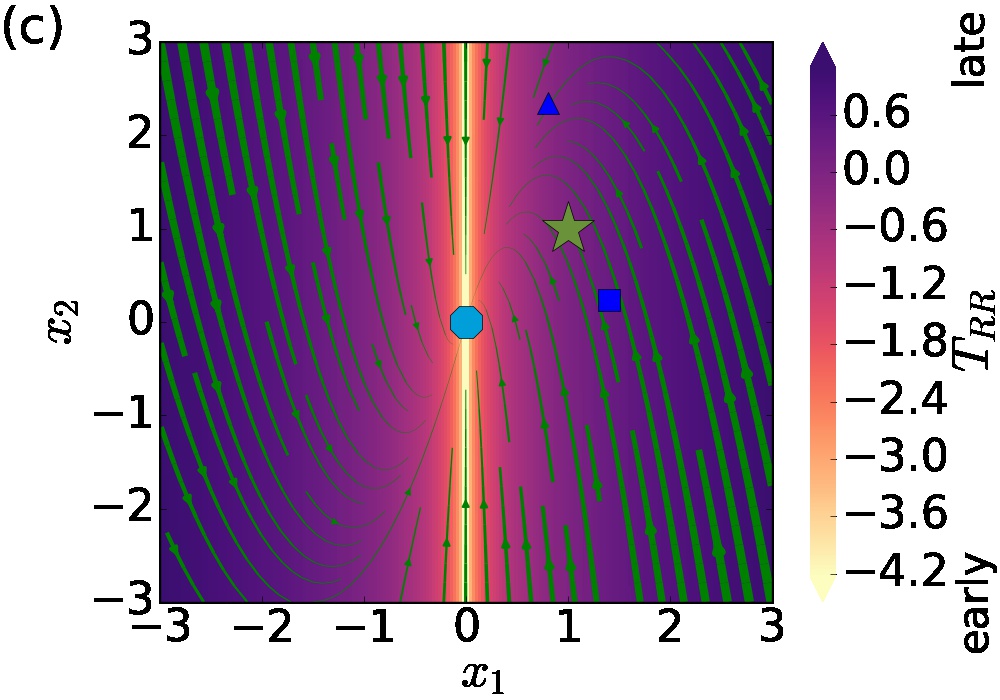}\label{fig:linear-nsp-rrt}} \\ 
			\vspace{\vpixdist}			
			\subfloat{\includegraphics[width = \elementwidth]{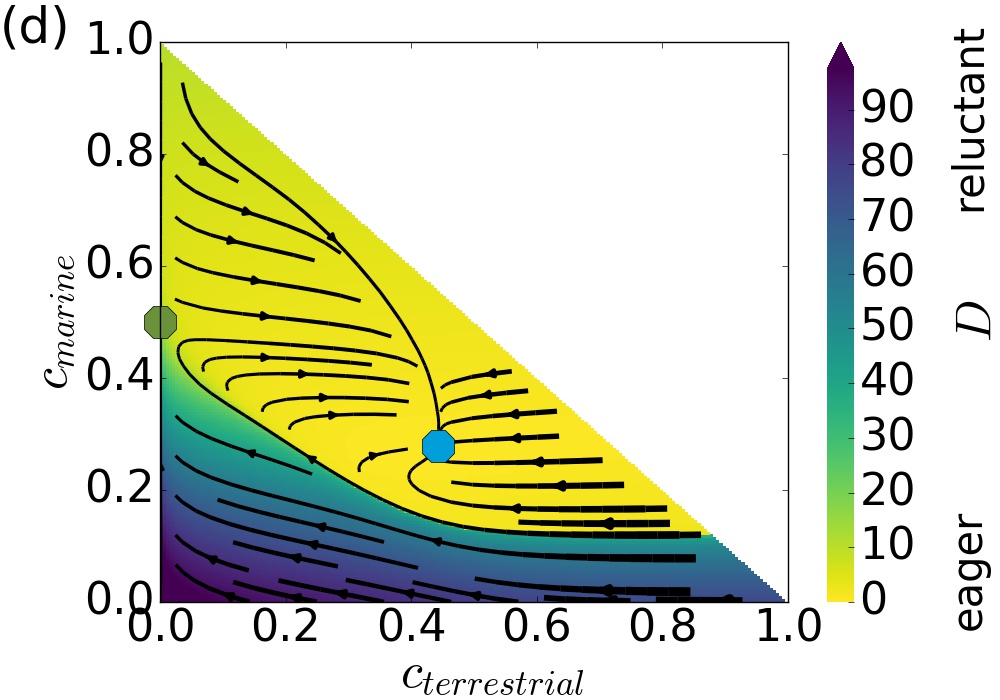}\label{fig:anderies-audic}} & 
			\subfloat{\includegraphics[width = \elementwidthtwo]{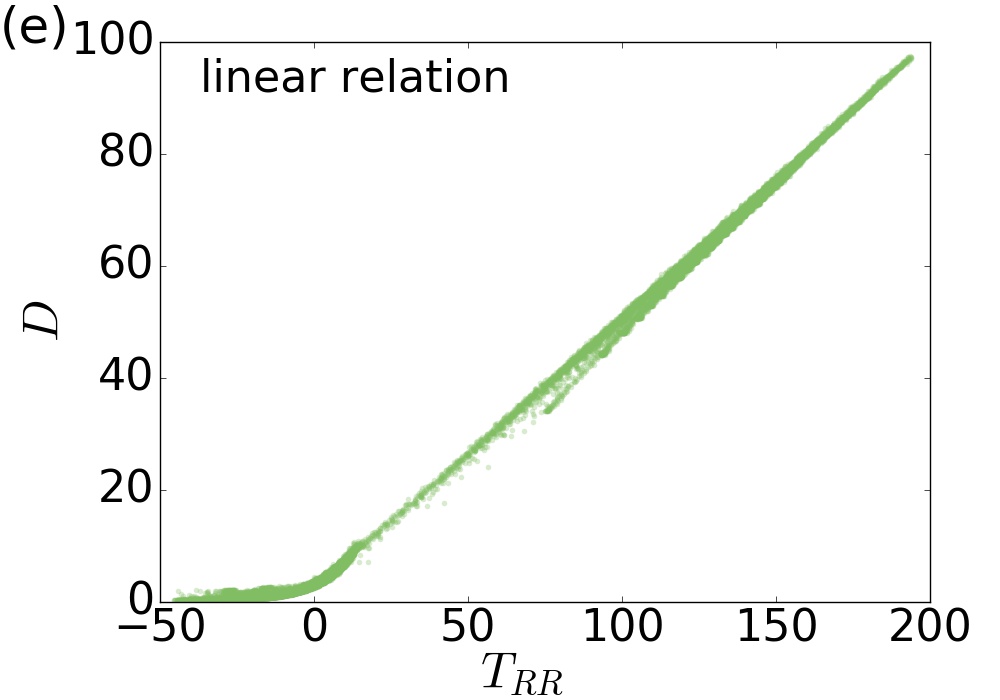}\label{fig:anderies-scatter}} &  
			\subfloat{\includegraphics[width = \elementwidth]{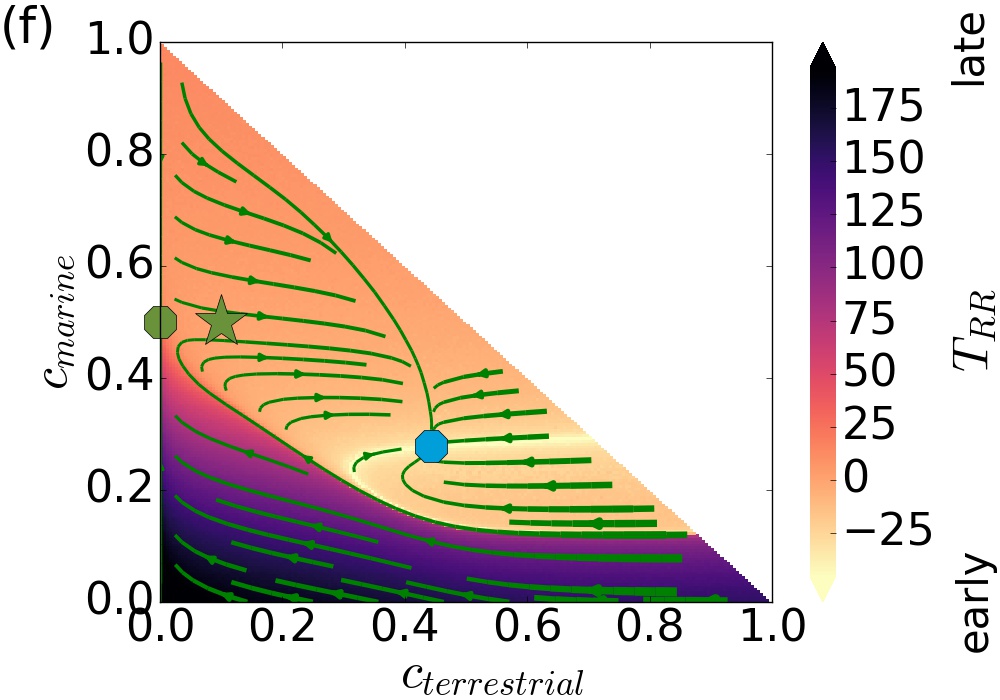}\label{fig:anderies-rrt}}  \\
			\subfloat{\includegraphics[width = \elementwidth]{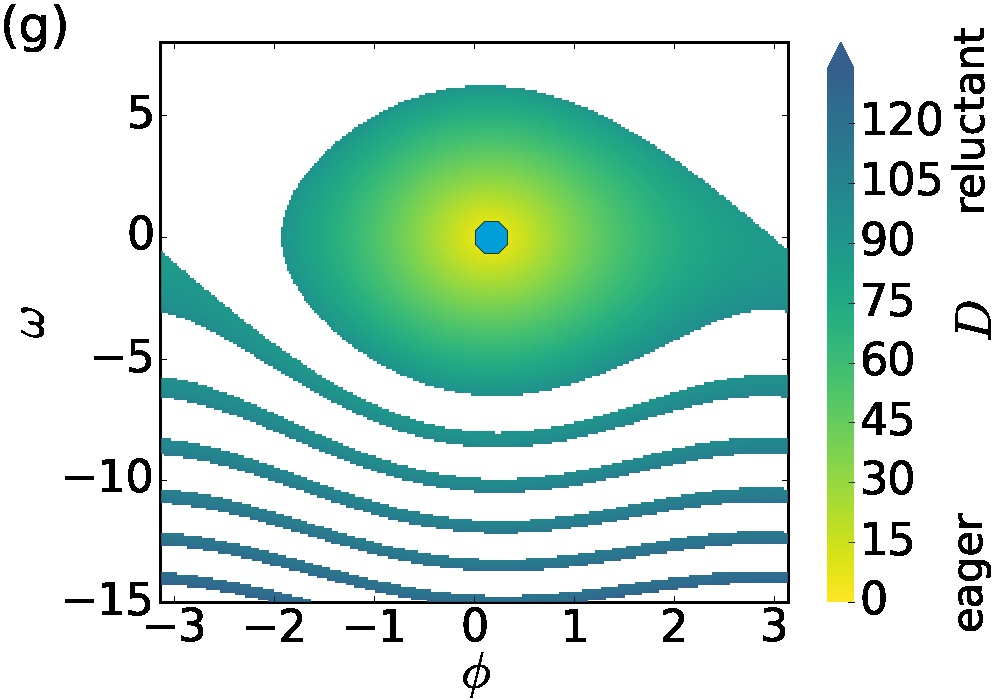}\label{fig:swing-audic}} &
			\subfloat{\includegraphics[width = \elementwidthtwo]{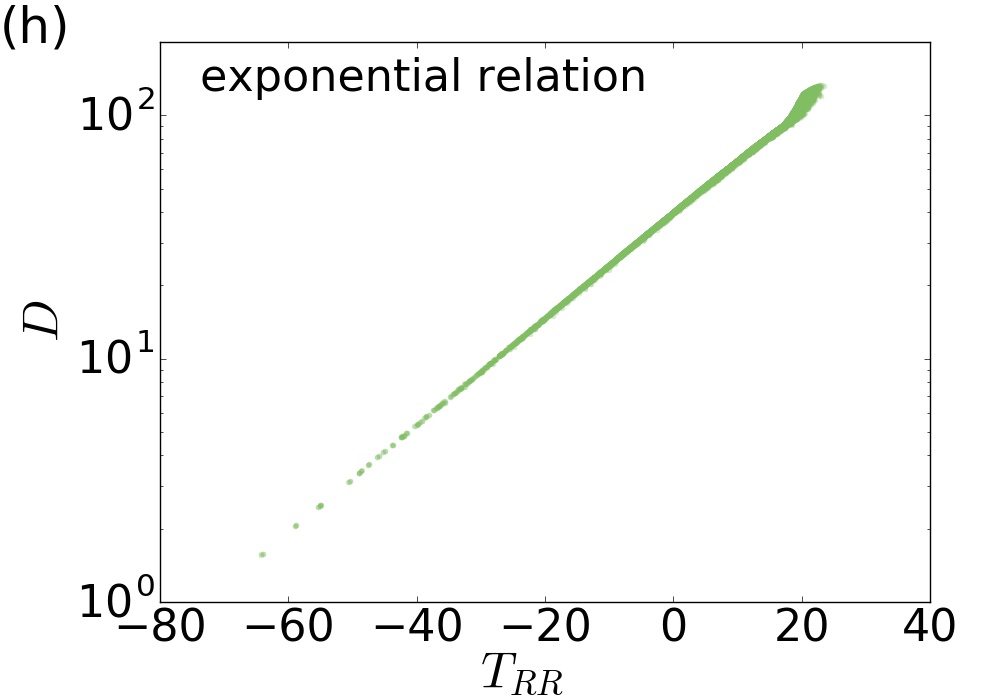}\label{fig:swing-scatter}} &
			\subfloat{\includegraphics[width = \elementwidth]{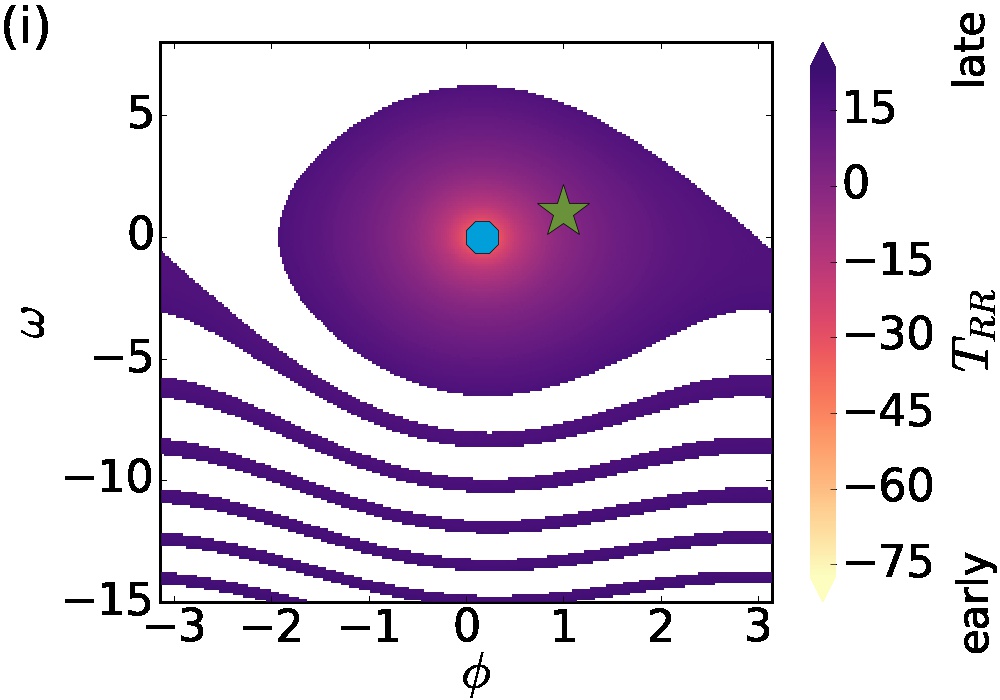}\label{fig:swing-rrt}} \\
			\subfloat{\includegraphics[width = \elementwidth]{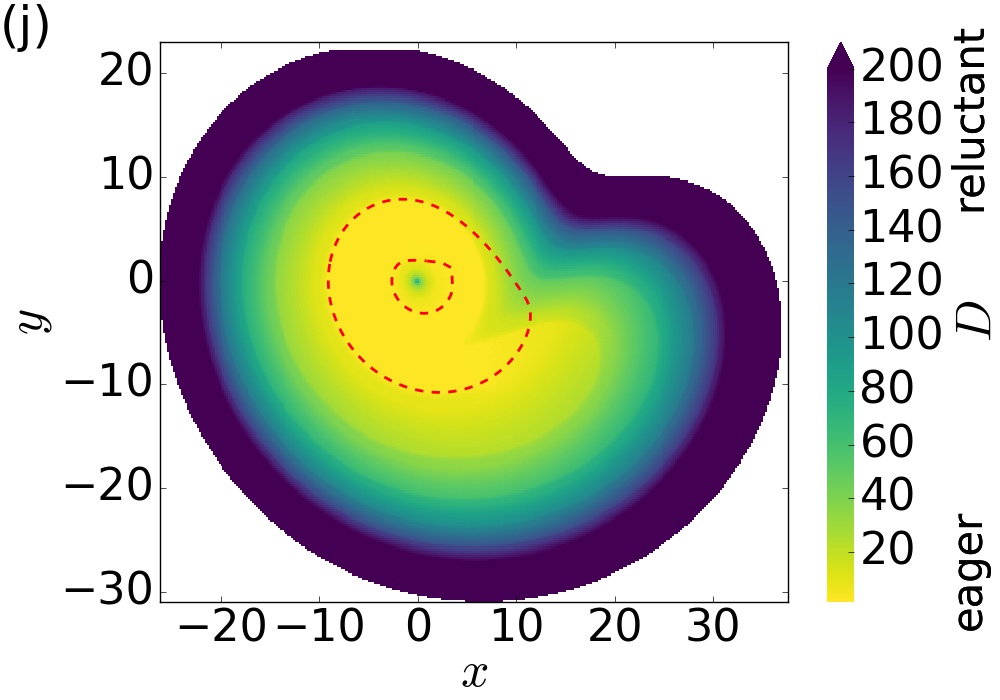}\label{fig:roessler-audic}} & 
			\subfloat{\includegraphics[width = \elementwidthtwo]{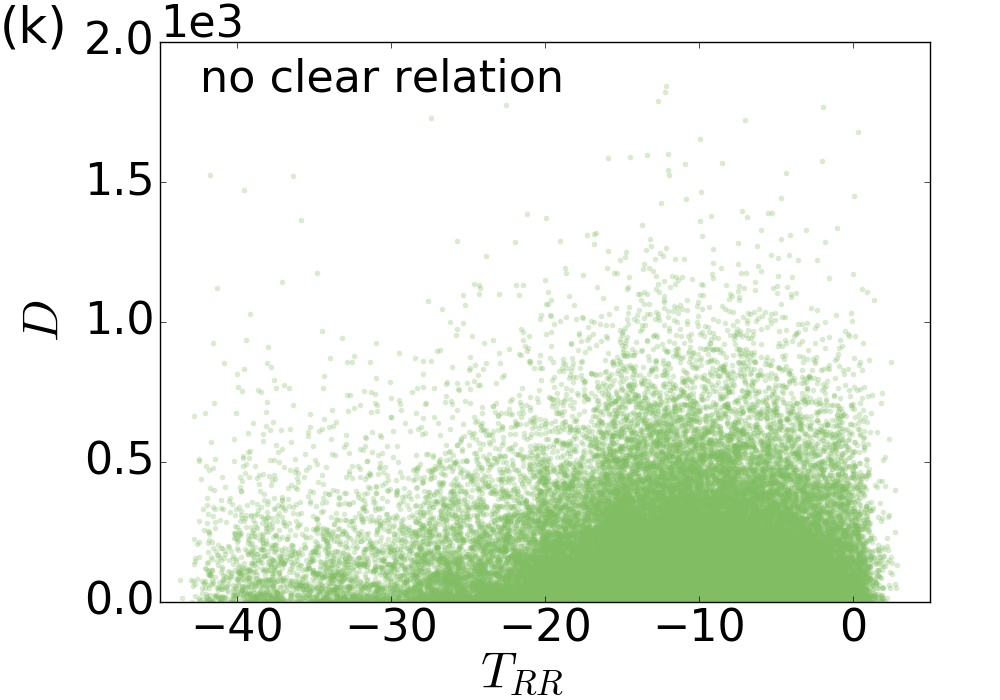}\label{fig:roessler-scatter}} & 
			\subfloat{\includegraphics[width = \elementwidth]{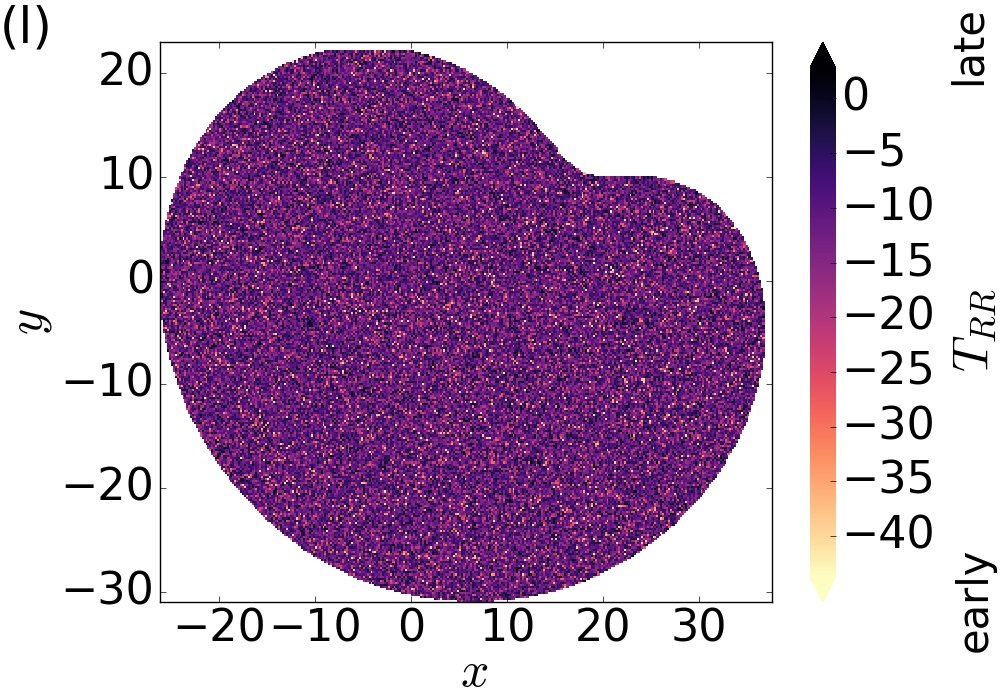}\label{fig:roessler-rrt}} 
		\end{tabular}
		\caption{(color online) For the presented example systems (top to bottom: linear system, global carbon cycle, generator in a power grid, Rössler system) the two new metrics have been computed for each initial condition in the state space and marked with color, see left column \emph{Area under Distance Curve} (\audic{}, $\audic{}$) and right column \emph{Regularized Reaching Time} (\rrt{}). The middle column shows their relations for the particular system. The initial conditions $x^a_0$ (triangle) and $x^b_0$ (square) from \Cref{fig:linear-example} have been marked in (a-c), too. Interestingly the manifold (red dashed line) where \audic{} increases the strongest is tilted with respect to the center manifold. As the Rössler system is 3-dimensional, the above plot depicts only a slice at fixed $z=0.6$. Furthermore, the dashed, red line marks the boundary of the attractor's projection to this plane.}
	\end{figure*}
	\egroup
	
\subsection{Linear system with two different time scales}
\label{sec:linear-system}

	The first example is the linear system in \Cref{eq:linear-system}. We chose dimension 2 in order to show the basic features while still being able to map the phase space. But the results can be applied in any dimension.
	\begin{align}
		\dot{x} = A\cdot x	\quad \text{with } A = \begin{pmatrix*}[r]
			-1 & 0 \\
			4 & -2
			\end{pmatrix*}\label{eq:linear-system}
	\end{align}
	
	While in general cases, \rrt{} and \audic{} can only be tackled numerically, we can solve the linear system analytically. The full details of this calculation are in \ref{app1-sec:linear-calculation} and the main results are (with $x_{ref} = (1, 1)^T$)
	\begin{align}
		\rrt{}(x) =& \ln\left(|x_1|\right) \label{eq:linear-rrt-result} \\
		D(x) =& \frac{1}{2}\left[ \frac{11}{3}x_1^2 + \frac{1}{2}x_2^2 + \frac{4}{3}x_1 x_2\right]. \label{eq:linear-audic-result}
	\end{align}
	While the result for \audic{} seems intuitive, \rrt{} might be more surprising. The result depends only on the $x_1$-axis. This is because the strong stable manifold, i.e.\ the one corresponding to the smaller Lyapunov exponent, is in the $x_2$-direction. But for $\epsilon \longrightarrow 0$ in \Cref{eq:def-rrt}, which implies $t \longrightarrow \infty$ for a trajectory, the contribution from the smaller Lyapunov exponent disappears. Thus, only the orthogonal part $x_1$ is relevant.
	
	In order to get a better feeling for these novels metrics, we have chosen two exemplary initial conditions, an \emph{early-eager} one and a \emph{late-eager} one, and plotted their trajectories' distance to the attractor over time in \Cref{fig:linear-example}. Thus, the blue-shaded area corresponds to the \audic{} value which is the same in both cases of our particular choice. In order to demonstrate the intuition that \rrt{} can be interpreted as the time-shift between the original trajectory and the reference trajectory until the asymptotics matched, we show both trajectories shifted to each other.
	
	The results from \Cref{eq:linear-rrt-result,eq:linear-audic-result} can also be seen in the numerical simulations in \Cref{fig:linear-nsp-rrt,fig:linear-nsp-audic}. The coloring describes the values of the metrics (cmp. the colorbar in the right of the figures) and green star represents the reference point for the \rrt{} computation. Note how the manifold where \audic{} increases the slowest (red dashed line in \Cref{fig:linear-nsp-audic}) is tilted with respect to the center manifold. This really proves that we are looking at transient phenomena and we have to take more into account than only the slowest dynamics.
	
	The exponential lower bound that comes up in the correlation diagram \Cref{fig:linear-nsp-scatter} can be calculated analytically (see \ref{app1-sec:linear-calculation})
	\begin{align}
		D(x) \geq  \frac{25}{18}e^{2\rrt{}}. \label{eq:linear-exponential-lower-bound}
	\end{align}
	
\subsection{Global carbon cycle}

	\begin{figure}
		\centering
		\includegraphics[width=\columnwidth]{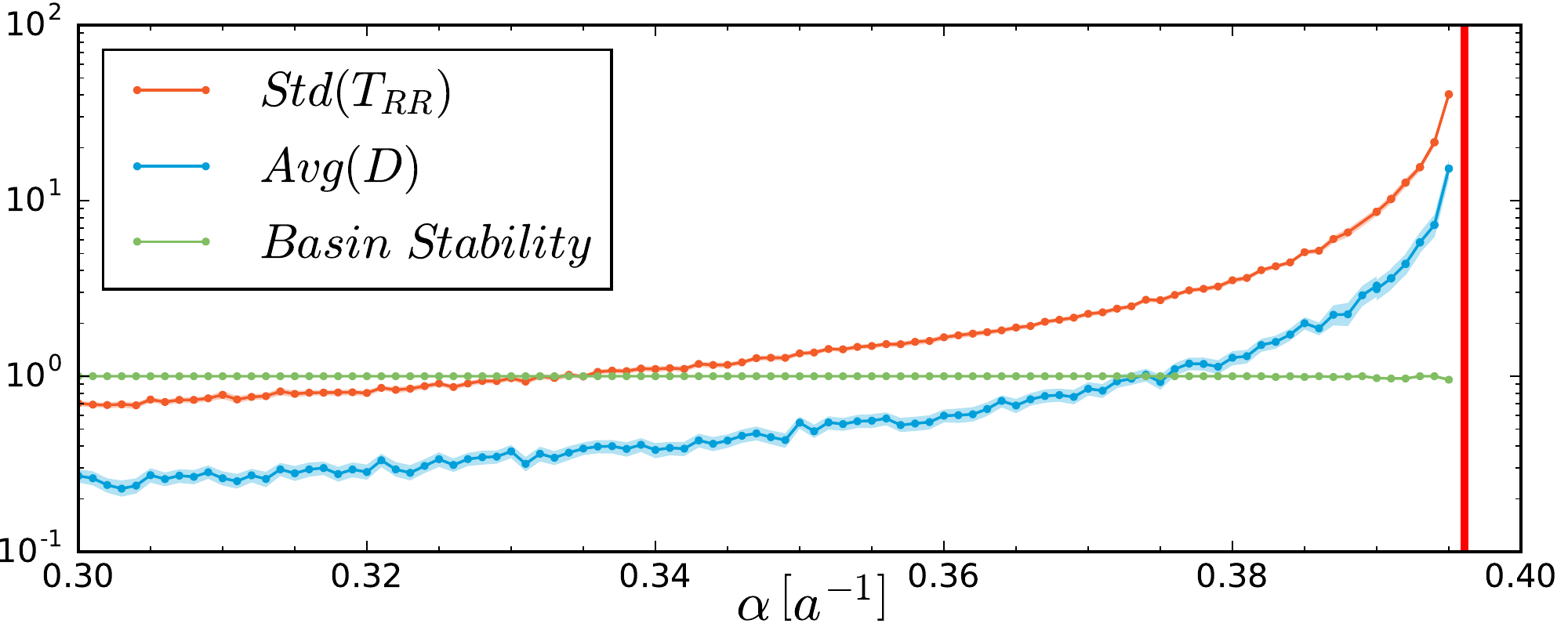}
		\caption{(color online) For the \emph{global carbon cycle} in \Cref{eq:gcc-ct,eq:gcc-cm}, the mean of \audic{} ($\audic{}$) and standard deviation of \rrt{} are plotted (with their 5\% and 95\% bootstrap errors) and show a divergence before the parameter $\alpha$ (yearly human carbon offtake) reaches the bifurcation value (marked by the red line). For comparison Basin Stability has been added which does \emph{not} show any change because the size of the basin stays constant before the bifurcation.}
		\label{fig:divAnderies}
	\end{figure}
	
	The second example is a conceptual model of the global carbon cycle proposed by Anderies~et~al.~\cite{anderies2013topology}, where we used the pre-industrialization version. It consists of three dynamical variables, the terrestrial, maritime and atmospheric carbon stocks denoted  $c_t$, $c_m$ and $c_a$ respectively, and the constraint $C = c_t + c_m + c_a = \ const$. Thus, we can reduce the system to 2 state variables $c_t$ and $c_m$ and rescale the units such that $C = 1$ arriving at
	\begin{subequations}
		\begin{align}
		\dot{c}_t =& NEP(p, r, c_t) - \alpha c_t \label{eq:gcc-ct}\\
		\dot{c}_m =& I(c_a, c_m), \label{eq:gcc-cm}
		\end{align}
	\end{subequations}
	where $NEP$ is the net eco-system production, $p$ photosynthesis, $r$ respiration, $\alpha$ harvesting parameter and $I$ diffusion; indirect dependencies have been omitted and more details are in \cite{anderies2013topology,heck2016esd}.
	
	The whole phase space of \Cref{eq:gcc-ct,eq:gcc-cm} as depicted in \Cref{fig:anderies-rrt} is the basin of the attraction of a fixed point in the middle marked by a blue dot; the dynamics is drawn as streams. Note that the trajectories starting in the lower part have to pass by a ``desert-like'' saddle (with $c_t = 0$) at the left (green dot). 
	
	The color in this graph depicts \rrt{} and the first finding is the splitting of the basin of attraction. Furthermore, the strong stable manifold becomes visible as a light beige line due to their low values of \rrt{}, i.e.\ as very \emph{early} states because \rrt{} diverges to $-\infty$. This proves it being the separatrix for the observed splitting and it will merge to an arm of the stable manifold corresponding to the saddle arising after the subcritical pitchfork bifurcation mentioned below. Also, the expected smooth increase of the return times when distancing (along the trajectories) from the attractor can be observed.
	
	When applying \audic{} to this model (\Cref{fig:anderies-audic}) the splitting of the basin can be observed again. In contrast to \rrt{}, the stable manifold is not visible because \audic{} can be seen as a (by distance) weighted time and the contributions for the asymptotic part where the difference in the Lyapunov spectrum matters are negligible. Furthermore, we see a clear linear correlation of both metrics in \Cref{fig:anderies-scatter} because all trajectories starting in the lower part have to pass by at the saddle on the left and spend a long time there.
	
	This example shows how saddles can induce long transients, as stressed by Hastings \cite{hastings2004transients}, and that our metrics react appropriately.
	
	It is important to note that the metrics are \emph{early-warning signals}, too. When increasing $\alpha$, corresponding to the harvest of terrestrial carbon, the system passes through a subcritical pitchfork bifurcation where the saddle becomes stable and the lower-left part of the phase space splits off. The divergences of the two metrics' statistics as seen in \Cref{fig:divAnderies} prove their prebifurcational sensitivity, while other systemic indicators like basin stability \cite{menck2013basin} do not change (up to numerical fluctuations, see \Cref{fig:divAnderies}).
	
\subsection{Generator in a power grid}
	\label{sec:swing}
	
	As an example of intermediate complexity, we chose the swing equations \cite{menck2014dead,yuan2003solution} in \Cref{eq:swing1,eq:swing2}, a model describing the dynamics of a single generator connected to a large power grid. It consists of two dynamical variables, the phase $\theta$ and angular frequency $\omega$, both in a reference frame rotating at the grid's rated frequency. The parameters of the system correspond to the net power production $P=1$ (at the node), the capacity of the transmission line $K = 6$ and dampening $\alpha = 0.1$.
	
	The stable fixed point at  $\omega_0 = 0$, $\phi_0 = \arcsin \frac{P}{K}$ describes a state of synchronization. For the chosen set of parameters, the system exhibits another attractor: a limit cycle at larger positive values of $\omega$. For negative values, the two basins of attraction are interleaved. A more detailed introduction and analysis can be found in \cite{menck2014dead,schultz2014detours}.
	\begin{subequations}
		\begin{align}
			\dot{\phi} &= \omega \label{eq:swing1} \\
			\dot{\omega} &= 2P - \alpha \omega - 2K \sin \phi \label{eq:swing2}
		\end{align}
	\end{subequations}
	
	Calculating \rrt{} inside the basin of the stable fixed pointed $(\omega_0, \theta_0)$ yields \Cref{fig:swing-rrt}. There is basically no color change away from the attractor, so we can see that a trajectory will barely spend any time in the transient and goes quickly to the attractor. Analogously, \Cref{fig:swing-audic} for \audic{} leads to a similar conclusion. 
	
	Comparing both metrics, see \Cref{fig:swing-scatter}, shows that they are closely linked. Note that $D$ is presented on a logarithmic scale, so the relation is exponential and can be explained using the calculations for a linear focus in \ref{app1-sec:linear-calculation}. So what we see here is actually the influence of the linearized part of the system.
	
	Note that roughly 30\% of the lower part of the phase space, where the nonlinearities actually have an influence, follow the exponential relation, too, because the transient is very fast and thus its influence is rather low.
	
	The aforementioned limit cycle corresponds to the system being far away from synchrony and generators would usually switch off before reaching it. As it is not so relevant, the treatment of this attractor is in \ref{app1-sec:swing-limit-cycle}.

\subsection{Chaotic Rössler oscillator}
	
	Although we have proven the convergence for fixed points only, we show with the chaotic Rössler system \cite{rossler1976equation,zgliczynski1997computer} that our metrics are applicable to higher-dimensional and more complex attractors also
	\begin{subequations}
	\begin{align}
		\dot{x} = & -y -z \label{eq:roessler1}\\
		\dot{y} = &\ x + ay \label{eq:roessler2} \\
		\dot{z} = &\ b + z\left( x- c \right).  \label{eq:roessler3}
	\end{align}
	\end{subequations}
	
	\Cref{fig:roessler-rrt} shows a slice of the phase space with the standard parameters $a = 0.2$, $b = 0.2$, $c = 5.7$ for \rrt{} and the expected sensitivity to initial conditions for chaos is observed: \emph{early} and \emph{late} trajectories lie closely together and the metric \rrt{} has low spatial correlation.
	
	In contrast, \audic{} shows in \Cref{fig:roessler-audic} surprisingly smooth changes of an embryo-like shape. Because the focus of this article is on transient dynamics a new feature of the chaotic Rössler system is uncovered: while the attractor is chaotic, the basin of attraction is very regular. \audic{} focuses on the initial transient and the chaotic asymptotics is filtered out. For comparison, the boundaries of the attractor's projection have been added with dashed red lines in \Cref{fig:roessler-audic} and depictions of the attractor are in \ref{app1-sec:roessler-attractor}.
	
	We can deduce that even though the system is chaotic the strong sensitivity to initial conditions happens rather late in the transient when the trajectory is already close to the attractor, because \rrt{} focuses more on the intermediate transient.
	
	However, this implies that \rrt{} can be successfully applied as an early-warning signal in this case, too. In order to demonstrate this, we chose to vary $a$ as it has a crucial influence on the system's dynamics (see the bifurcation diagram in \Cref{fig:roessler-bifurc} (green)). While for values of $a < 0.006$ (cf. \cite{barrio2011qualitative}) there is only a single stable fixed point, at $a \approx 0.006$ a limit cycle emerges due to a Hopf bifurcation \cite{barrio2011qualitative}. For $a > 0.11$, several period doublings are observed, finally leading to chaos for $a > 0.155$. Even in the chaotic regime, further bifurcations can be observed.
	
	\setlength{\elementwidth}{\columnwidth}
	\begin{figure}
		\centering
		\includegraphics[width = \elementwidth]{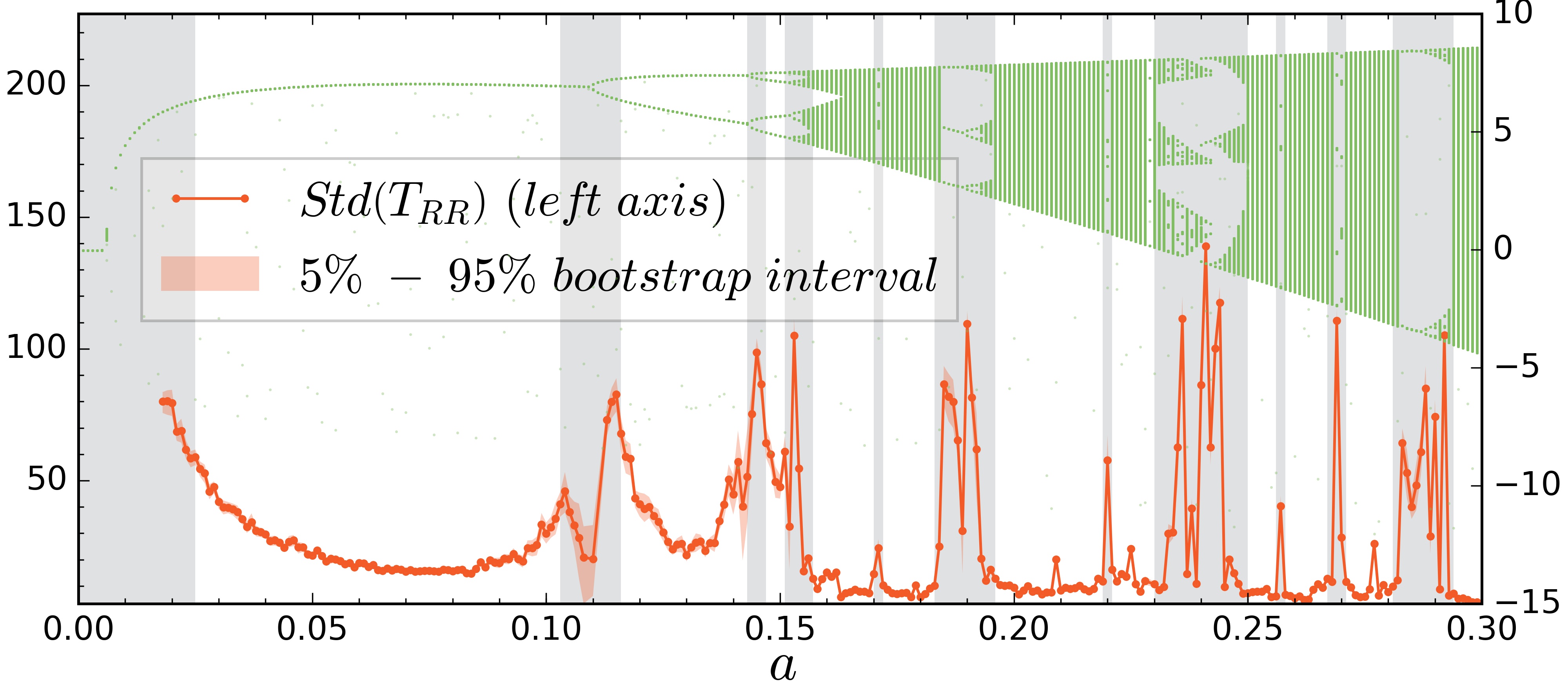}
		\caption{(color online) The bifurcation diagram (green) of the Rössler system for varying the parameter $a$ in \Cref{eq:roessler1,eq:roessler2,eq:roessler3} was computed from the local maxima in $z$ of the attractor and \rrt{} (orange) shows a strong sensitivity to these qualitative changes. The gray background is used so the reader can more easily connect the peaks in \rrt{} to the corresponding parts in the bifurcation diagram.}
		\label{fig:roessler-bifurc}
	\end{figure}
	
	In \Cref{fig:roessler-bifurc}, the standard deviation of the \rrt{} distribution from randomly chosen initial conditions inside the basin of attraction is given. Due to the sensitive dependence on initial conditions, the reference value varies a lot and hence introduce shifts in the distribution that do not describe actual changes in the system's dynamics. To remove this effect, it is crucial to use central statistics like the standard deviation.
	
	\rrt{} is strongly sensitive to any qualitative changes in the dynamics of the system, incl.\ even chaos-chaos transitions. Closely observing \Cref{fig:roessler-bifurc} uncovers that there is a base-line with a little noise at $\rrt \approx 10$ complemented with strong peaks. In the chaotic regime, the peaks correspond directly to qualitative changes. Also, we observe sensible changes during the period-doubling phase and a strong increase before the Hopf bifurcation at $a \approx 0.006$, proving the usefulness as an \emph{early-warning signal}.
	
	The abrupt downward peak at $a \approx 0.11$ is unexpected and more details are needed to clarify it.
	
	Note that sometimes the result of estimating the limit numerically fluctuated slightly, suggesting that an improved definition of \rrt{} could be sensible. The defining property could be the constant negative orbital derivative $\left.\frac{d}{d\tau}\rrt(x(\tau))\right|_{\tau=0} = -1$ with $\rrt (y) = 0\quad \forall y \in \mathcal{F}^{ss}\left(x_{ref}\right)$, where $\mathcal{F}^{ss}\left(x_{ref}\right)$ is the leaf of the strong stable foliation containing the reference point $x_{ref}$. This turns out to be highly non-trivial and is an issue for future research.
	
\section{Discussion and outlook}
\label{sec:discussion}

%
%

%
	
	We have treated problems (I)-(IV) arising from common methods of analyzing the transient time needed to reach an attractor in nonlinear systems by introducing two complementing metrics: \emph{Area under Distance Curve} and \emph{Regularized Reaching Time}, and applied their properties, in particular them being Lyapunov functions, to the development of an efficient estimation algorithm. Furthermore, in \ref{app2-sec:definitions-and-theorems} we prove under mild conditions the existence of these functions.
	
	In order to show their applicability and usefulness, we chose 4 example systems and analyzed them with respect to these novel metrics.
	
	In the linear system everything could be solved analytically and we found the tilt of the slowest increasing \audic{} axis, demonstrating that even this simple system already can have a rich transient. Furthermore, we explained why the strong-stable manifold turns out to be a singularity of the system's \rrt{} function but could argue that this is irrelevant for the estimation of statistics as the integral is still finite.
	
	The global carbon cycle demonstrated the importance of the transient analysis, as the desert state is only a saddle but nevertheless passing by there would lead to an extinction of humanity. The artificially estimated splitting of the basin in \cite{anderies2013topology} arises naturally here and \rrt{} uncovers that the separatrix is the strong-stable manifold of the system. Furthermore, the saddle induces an unexpected linear relation between \audic{} and \rrt{}. Particularly interesting is how the (central) statistics of our metrics are a systemic approach to the concept of \emph{critical slowing down} \cite{scheffer2009early,scheffer2012anticipating,lenton2011early} and react strongly to prebifurcational changes. Hence they are \emph{early-warning signals} for fundamental changes of the system. 
	
	The generator in a power grid displays how even a very nonlinear system can have rather unimportant transients while the asymptotics is mostly relevant. Thus our measures show the exponential relation expected from a linear focus.
	
	In order to prove the applicability to more complex dynamics, we used our metrics on the Rössler system, too, and found the smoothness of the attractor's basin with \audic{}. As the attractor itself is rather chaotic, this smoothness is surprising. Even though \rrt{} reacts strongly to the sensitivity to initial conditions of the chaotic system its worth is displayed when varying the $a$ parameter. This parameter has strong influence on the Rössler system's dynamics and \rrt{} reacts strongly to the different bifurcations and even the chaos-chaos transitions, proving again its worth as \emph{early-warning-signal}.
	
	Furthermore, using \audic{} as a cost or damage functions can be applied in the context of earth system analysis and climate impacts. In the future, applying this metric for estimating viable pathways without transgressing e.g. the Planetary Boundaries \cite{rockstrom2009planetary,steffen2015planetary} is one goal of this work.
	
	We did not perform any comparative analysis with the mentioned ``$\epsilon \longrightarrow 0$''-approaches because these behave inconsistently and their quantitative results are arbitrary, as discussed in length in \Cref{sec:problems-reaching-time} and \ref{app1-sec:divergence-problem}.
	
	Four directions of immediate future research are due:
	
	(1) Improving the definition of \rrt{} further using the Lyapunov function properties as described above. In order to generalize these definitions to even more complex systems this step seems crucial.
	
	(2) Applying the metrics to more complex systems to understand them and their properties to topological structures, e.g. in complex networks \cite{havlin2012challenges}, in more detail.
	
	(3) Introducing more sophisticated methods of Lyapunov function estimations \cite{giesl2015review}. The curse of dimensionality is going to be a problem for network systems, hence methods for estimation these metrics statistics in such kinds of systems will need novel algorithms.
	
	(4) Comparison of the timing of transients in model output and observation data. With the observable time, new possibilities for comparison with observation data are available and should be used.

%% file: Kittel-TimingOfTransients-appendix1.tex
\section{Details on the problems of reaching time definitions}
\label{app1-sec:divergence-problem}

\begin{figure*}	
	\includegraphics[width = \textwidth]{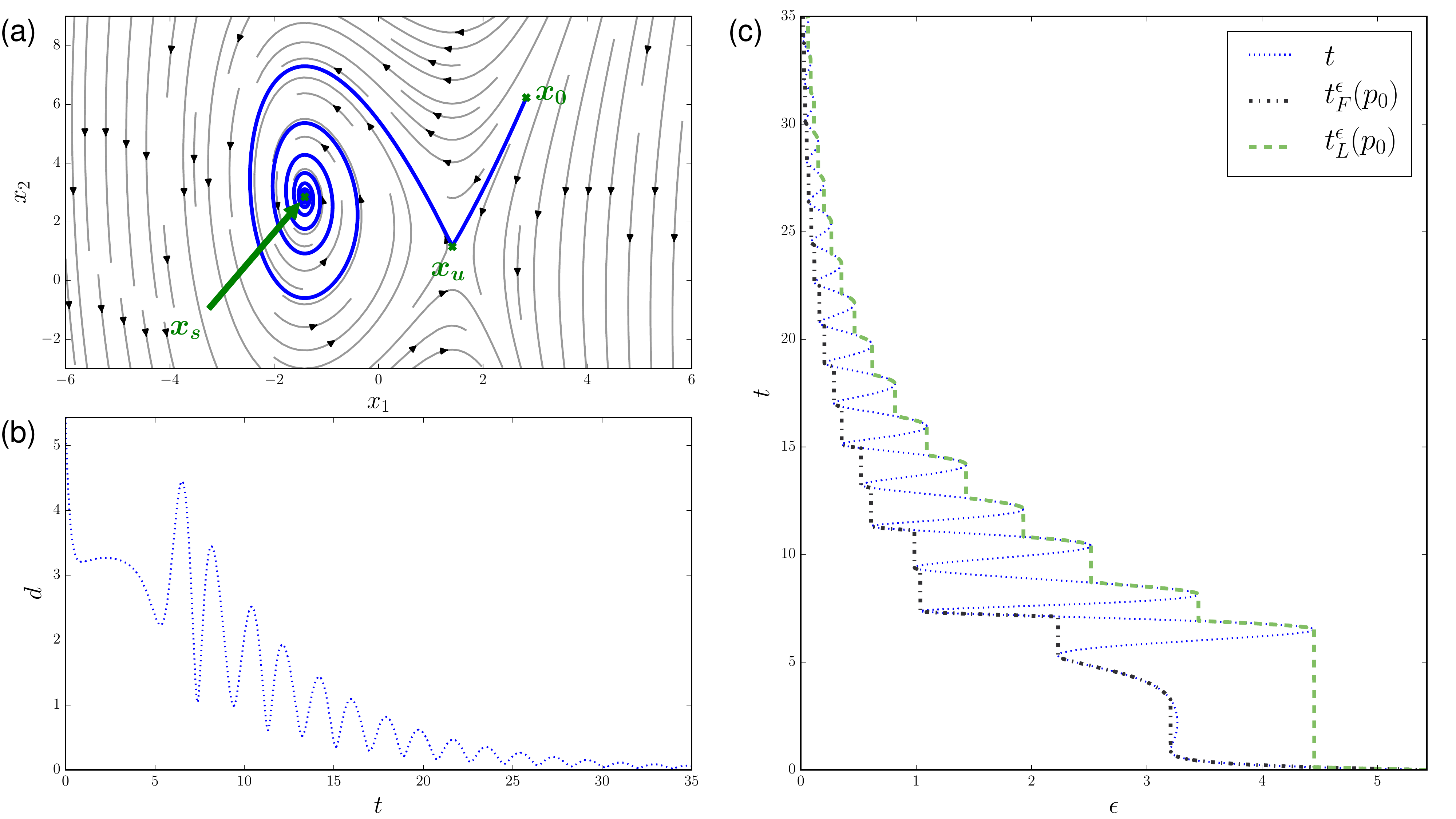}
		\subfloat{\label{app1-fig:ex-phase-space}}
		\subfloat{\label{app1-fig:ex-dist}}
		\subfloat{\label{app1-fig:eps-t}}
		\vspace{-4mm} 
	\caption{(color online) The phase space for the example system \Cref{app1-eq:example-odes1,app1-eq:example-odes2} using the parameters $a = 2$ and $b = 0.3$ is depicted in (a). Furthermore, the stable spiraling node $x_s$ and the saddle $x_u$ are added. The trajectory (blue) starting at $x_0$ closely passes by $x_u$ before it finally circulates in to $x_s$.
	(b) shows the (Euclidean) distance $d$ (dotted, blue) of this trajectory to its attractor $x_s$ over time $t$. The first longer dip between $t = 1$ and $t = 5$ is the transient at the (unstable) saddle $x_u$ while the oscillations afterwards are the spiraling around $x_s$.
	(c) turns (b) around in order to show the dependence of the time $t$ on some distance $d = \epsilon$ (dotted, blue) of the trajectory to the attractor $x_s$. Secondly, there are multiple values of $\epsilon$ for each $t$ so observables like $t^\epsilon_1(x_0)$ and $t^\epsilon_l(x_0)$ need to be introduced. $t^\epsilon_1(x_0)$ (dash-dotted, black) marks when the first time the $\epsilon$-neighborhood around $x_s$ is entered and $t^\epsilon_l(x_0)$ (dashed, green) the last time. The implications, particularly the arising problems for time definitions, are described in text.
}
\end{figure*}
To understand the problems of known \emph{reaching time definitions} in more detail, we introduce a small example system (\Cref{app1-eq:example-odes1,app1-eq:example-odes2}) containing a saddle-node-bifurcation at $a = 0$ followed by a node-focus transition at $a = \frac{b^4}{64}$. Hence, for $a > \frac{b^4}{64}$, a flow like the one depicted in \Cref{app1-fig:ex-phase-space} with a stable fixed point at $x_s = (x_{1,s}, x_{2,s}) = \left(-\sqrt{a}, 2( 1 + b \sqrt{a} )\right)$ and a saddle at $x_u = (x_{1,u}, x_{2,u}) = \left(\sqrt{a}, 2( 1 - b \sqrt{a} )\right)$ is obtained. Note that this system was mainly chosen in order to present all occurring problems with just a single example. 
\begin{subequations}
	\begin{align}
	\dot{x_1} =& 1 - \frac{x_2}{2} - bx_1 \label{app1-eq:example-odes1} \\
	\dot{x_2} =& 2\left(a - x_1^2\right) \label{app1-eq:example-odes2}
	\end{align}
\end{subequations}

Problem (I), the divergences, can be observed when plotting the distance $d$ of the example trajectory starting at some initial condition $x_0$ in \Cref{app1-fig:ex-phase-space} to the stable fixed point $x_s$ over time to get \Cref{app1-fig:ex-dist}. While $d$ approaches $0$ because $x_s$ is an attractor, it will reach the attractor in infinite time only, i.e. the \emph{reaching time} diverges. 

These problems are often addressed by measuring the time until the trajectory is ``close'' to the attractor. The relevant metrics would be the times when the trajectory enters an $\epsilon$-neighborhood around $x_s$ the first time ($t^\epsilon_F(x_0)$) and the last time ($t^\epsilon_L(x_0)$), i.e.\ the neighborhood is not left anymore. To illustrate this, we turned \Cref{app1-fig:ex-dist} around and added these two metrics in \Cref{app1-fig:eps-t}.

While the values for each of these metrics are now finite, their strong dependence on the choice of $\epsilon$ makes the physical interpretation (problem (II)) hard. Particularly, for the relevant small values of $\epsilon$ the times can become become as large as one wants due to their divergence for $\epsilon \longrightarrow 0$. Furthermore, it is even hard to state when the trajectory is ``close'' or still in the ``transient''. From \Cref{app1-fig:eps-t}, the discontinuities of the metrics (problem (III)) are observed, too. They go hand in hand with the physical interpretation problem as small changes in the choice of $\epsilon$ can induce jumps in the values of the time.

Problem (IV), non-invariance, can be understood by the fact, that a distance function is not invariant under a change of variables. This means, we would get different values for $t^\epsilon_F(x_0)$ and $t^\epsilon_L(x_0)$ if we change the variables. This principle can also be understood from a different point of view: Which is the correct function that one should choose to measure the distance to the attractor? This might sound trivial, but it is crucial for the results. As only values close to the attractor matter, but the metrics are very sensible to changes (problem (II)), the impact of this choice may become large.

		
\section{Analytical solutions of a linear system}
\label{app1-sec:linear-calculation}

\subsection{Area under Distance Curve}

For a linear system, we chose dimension 2 in order to show the basic features while still being able to map the phase space. But the general procedure can be applied in any dimension.

The systems is
\begin{align}
\dot{x} = A\cdot x
\end{align}

where $x, \dot{x} \in \R^2$ and $A \in R^{2\times2}$. This implies a fixed point at $x_* = (0,0)^T$ (that we will abbreviate with $0$ from now on) and we want $A$ to be (complex) diagonalizable and with negative real parts of all eigenvalues. Hence $0$ is exponentially stable.

In this part we discuss the case of a real leading eigenvalue and the case with complex leading eigenvalues can be treated similarly while keeping Proposition \ref{app2-prop:RRTnorm} in mind. For reasons of simplicity we use as distance function $d(x, 0) = {x_1}^2 + {x_2}^2$. (Note that this is not a distance in mathematical terms but still captures how close on is to the attractor. This has been discussed in the main text and in \ref{app1-sec:conv-audic}.)

An ansatz for \audic{} is $D(x) = \frac{1}{2} x^T C x$ with a symmetric matrix $C\in \R^{2\times2}$. Plugging this in the orbital derivative equation for \audic{} leads to
\begin{align}
x^T C A x = -x^T x.
\end{align}
This implies that the asymmetric part of $CA$ equals negative unity. As this should be true for all $A$ (including non-symmetric ones) the negative inverse of $A$ is in general not the solution for $C$. Instead we note that for any asymmetric matrix $Q\in R^{2\times2}$, $d(x,0)$ can be written as
\begin{align}
d(x, 0) =& x^T x =x^T (\id + Q) x
\end{align}
because $x^T Q x = 0$.

Thus, we shifted the problem to finding an asymmetric matrix $Q$ s.t. $(\id + Q)A\inv$ is symmetric so we have
\begin{align}
C = -(\id + Q)A\inv. \label{app1-eq:C-from-Q-and-A}
\end{align}
To solve this, we introduce $d(d-1)$ parameters for the upper-right triangle of $Q$ (because the lower-left one is then defined by the constraint $Q_{ij} = - Q_{ji}$).

In the case of a 2-dimensional system, this is one parameter $q$:
\begin{align}
Q =& \begin{pmatrix*}[r]
0 & q \\
-q & 0
\end{pmatrix*} \label{app1-eq:Q-def}
\end{align}
Plugging \Cref{app1-eq:C-from-Q-and-A} in the symmetry constraint $C_{ij} = C_{ji}$ leads to
\begin{align}
-A_{01} + qA_{00} = -A_{10} -q A_{11}.
\end{align}
This can be solved for $q$
\begin{align}
q = \frac{A_{01} - A_{10}}{A_{00} + A_{11}} \label{app1-eq:result-q}
\end{align}
and by plugging $q$ back in \Cref{app1-eq:C-from-Q-and-A,app1-eq:Q-def}, a rather lengthy expression for $C$ can be obtained.

Choosing the matrix $A$ as in the example \Cref{eq:linear-system} yields the correct \audic{} function:
\begin{align}
A = \begin{pmatrix*}[r]
-1 & 0 \\
4 & -2
\end{pmatrix*} \quad \quad \quad D(x) = \frac{1}{2}x^T \begin{pmatrix*}[r]
\frac{11}{3} &  \frac{2}{3} \\
\frac{2}{3} & \frac{1}{2}
\end{pmatrix*} x. \label{app1-eq:linear-audic-result}
\end{align}

\subsection{Regularized Reaching Time}

Calculating \rrt{} in the linear case can be done, too. The first step is to chose an eigenvector basis $\lbrace v_i \rbrace_{i = 0, \dots d-1}$ corresponding to the eigenvalues ordered as $\Re\lambda_0 > \Re\lambda_1 >\dots$ of $A$ and find the (unique) decomposition of $x$ in that basis:
\begin{align}
x = \sum_i \alpha_i v_i. \label{app1-eq:decomposition-in-eigenvectors}
\end{align}
We can plug this in the solution of the linear system:
\begin{align}
x(t) = e^{At}x = \sum_i \alpha_i v_1 e^{\lambda_i t}.
\end{align}

Taking the limit as in \Cref{app1-eq:def-rrt} looks at large values of $t$ s.t. the terms for $i\geq1$ can be omitted, because the $\lambda_i$ have been ordered in an appropriate way:
\begin{align}
x(t) = e^{At}x_0 = \alpha_0 v_0 e^{\lambda_0 t} \quad \quad \text{for large }t.
\end{align}
For a fixed distance value $\epsilon = \left\| x(t) \right\|$ this equation can be solved for $t$ yielding the reaching time:
\begin{align}
t_{reach}(\epsilon, x_0) = \frac{1}{\lambda_0} \ln\left(\frac{\epsilon}{\| \alpha_0 v_0 \|}\right)
\end{align}
Plugging this in the definition for \rrt{} (using a reference point $x_{ref}$ with the decomposition coefficient $\alpha_{ref}$) gives
\begin{align}
\rrt{}(x) =& t_{reach}(\epsilon, x_0) - t_{reach}(\epsilon, x_{ref}) \\
=& \frac{1}{\lambda_0} \ln\left(\frac{|\alpha_{ref}|}{|\alpha_0|}\right). \label{app1-eq:linear-rrt-result}
\end{align}
Note that we did not use the \audic{} level sets as described in the main text. This could be done here but only lengthens the calculation while the result stays the same.

For the example in \Cref{eq:linear-system} the matrix 
\begin{align}
A = \begin{pmatrix*}[r]
-1 & 0 \\
4 & -2
\end{pmatrix*}
\end{align}
has been used, which has as eigenvalues and -vectors
\begin{align}
\lambda_0 = -1,\ 
v_0 =  \left(\begin{matrix}
1  \\
1 
\end{matrix}\right) 
\quad \text{and} \quad
\lambda_1 = -2,\ 
v_0 =  \left(\begin{matrix}
0  \\
1 
\end{matrix}\right) .
\end{align}
The decomposition of \Cref{app1-eq:decomposition-in-eigenvectors} is
\begin{align}
x = \left(\begin{matrix}
x_1  \\
x_2 
\end{matrix}\right) 
= x_1 v_0 + (x_2 - x_1) v_1
\end{align}
Thus $\alpha_0 = x_1$. Furthermore, we chose $x_{ref} = \left(\begin{matrix}
1  \\
1 
\end{matrix}\right)$ so $\alpha_{ref} = 1$. Plugging these results in \Cref{app1-eq:linear-rrt-result} yields the correct \emph{Regularized Reaching Time}
\begin{align}
\rrt{}(x) =& \ln\left(|x_1|\right). \label{app1-eq:linear-result-numerical}
\end{align}

In order to understand the exponential lower bound seen in \Cref{fig:linear-nsp-scatter}, we expand \Cref{app1-eq:linear-audic-result}.
\begin{align}
D(x) = \frac{1}{2}\left[ \frac{11}{3}x_1^2 + \frac{1}{2}x_2^2 + \frac{4}{3}x_1 x_2\right].
\end{align}

Because \rrt{} in \Cref{app1-eq:linear-result-numerical} depends only on $x_1$ we need to find the lower bound of $D$ for fixed $x_1$, i.e.~the minimum of the polynomial in $x_2$. We find $x_{2,min} =-\frac{4}{3}x_1$ and thus
\begin{align}
D(x) \geq  \frac{25}{18}x_1^2.
\end{align}
Taking the result from \Cref{app1-eq:linear-result-numerical} and plugging it in gives rise to the observed exponential bound
\begin{align}
D(x) \geq  \frac{25}{18}e^{2\rrt{}}. \label{app1-eq:D-geq-expTRR}
\end{align}

\section{Reaching time analysis of the limit cycle in the swing equation}
\label{app1-sec:swing-limit-cycle}

\bgroup
\setlength{\elementwidth}{0.28\textwidth}
\setlength{\elementwidthtwo}{\elementwidth}
\setlength{\tabcolsep}{0pt}
\setlength\vpixdist{-4.4mm}
\newcolumntype{C}[1]{>{\centering\let\newline\\\arraybackslash\hspace{0pt}\vspace{\vpixdist}}m{#1}}
\renewcommand*{\arraystretch}{0.2}
\begin{figure*}
	\centering
	\begin{tabular}{c c c}			
		\subfloat{\includegraphics[width = \elementwidth]{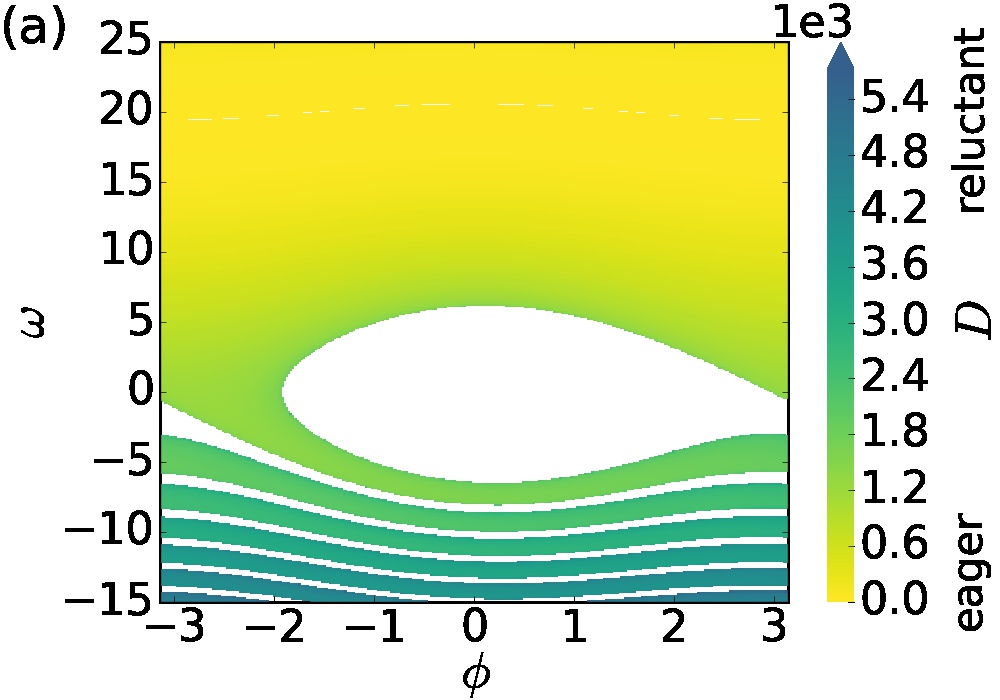}\label{app1-fig:swing-audic2}} &
		\subfloat{\includegraphics[width = \elementwidthtwo]{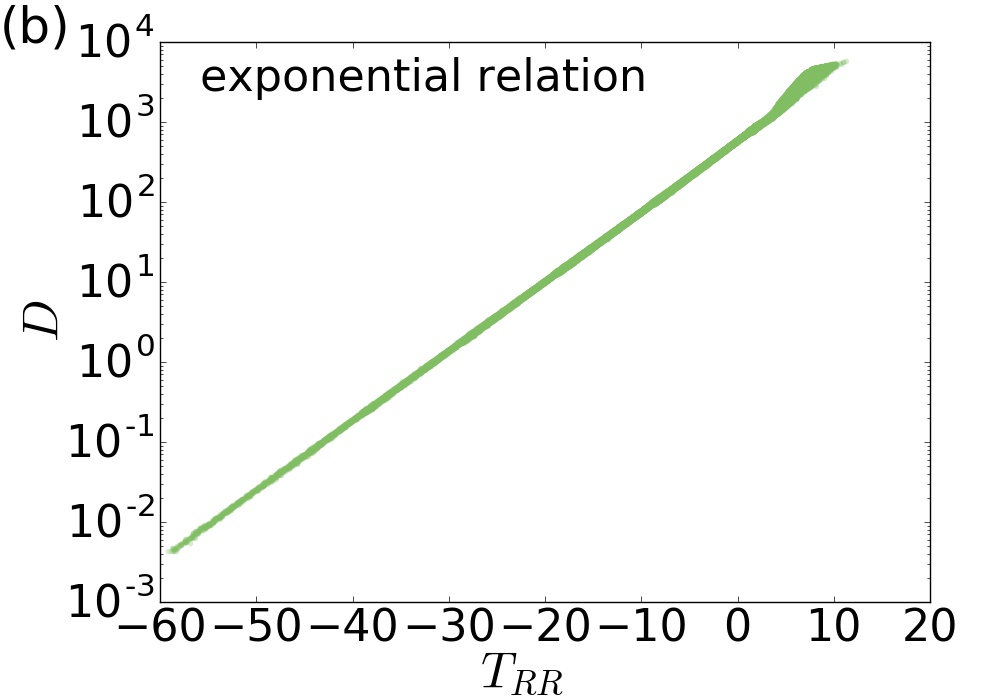}\label{app1-fig:swing-scatter2}}	&
		\subfloat{\includegraphics[width = \elementwidth]{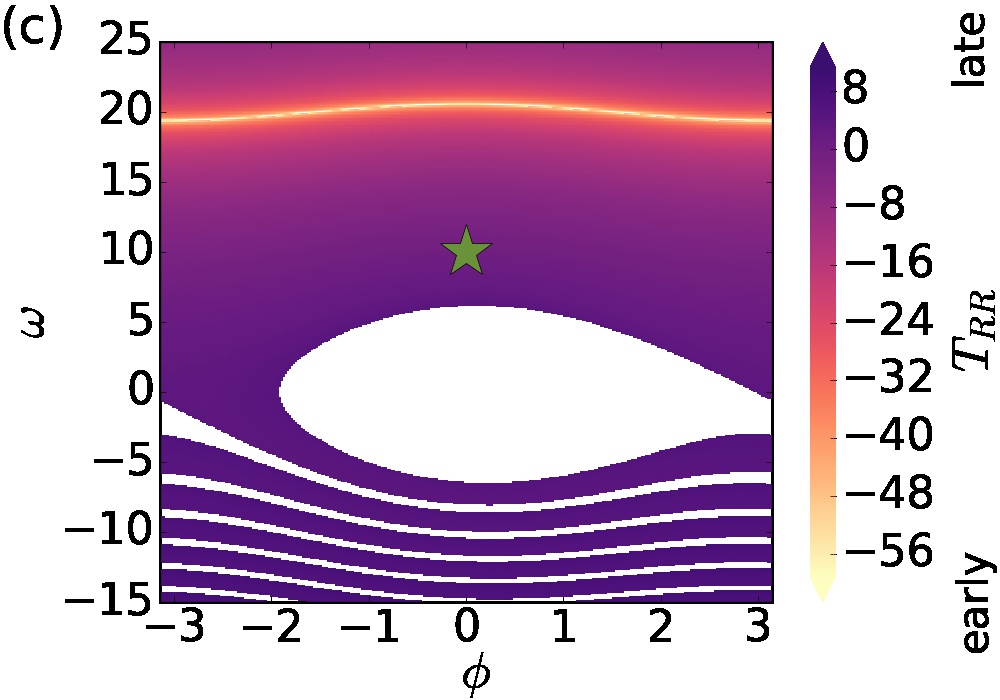}\label{app1-fig:swing-rrt2}}  		
	\end{tabular}
	\caption{(color online) The analysis of (a) \emph{Area under Distance Curve}, (c) \emph{Regularized Reaching Time} and (b) their relation for the limit cycle in the swing equation corresponding to the system ending up far away from synchrony.}
\end{figure*}
\egroup

As mentioned in \Cref{sec:swing}, the swing equations (\Cref{eq:swing1,eq:swing2}) have for the chosen set of parameters a limit cycle as the second attractor for values around $\omega \approx 20$ (see light beige line in \Cref{app1-fig:swing-rrt2}). This attractor corresponds to the generator being far from synchrony with the grid and usually it would have been switched off long before reaching it, so it is rather irrelevant and we analyze it here for completeness only.

Again, as with the fixed point attractor discussed in \Cref{sec:swing}, we observe in \Cref{app1-fig:swing-audic2,app1-fig:swing-rrt2} that the influence of the transient is rather low and the system approaches the attractor exponentially. Looking at Poincaré maps for fixed $\phi$ would give asymptotically exponential behavior. Thus the exponential relation in \Cref{app1-fig:swing-scatter2} can be explained with the calculations for the linear system in \ref{app1-sec:linear-calculation}.

\section{One-dimensional Systems}

\subsection{Closed Formulas}	
\label{app1-sec:1d-systems}

In one-dimensional systems
\begin{align}
\dot{x} = f(x)
\end{align}
with a stable fixed point at $x^*$, closed formulae for the \rrt{} and \audic{} can be written down. By separation of variables we get
\begin{align}
\int_{x_0}^{x^* \pm \epsilon}\frac{dx}{f(x)} = t(x_0, \epsilon),
\end{align}
where $x_0$ is initial state, $\epsilon$ the difference to the fixed point, $\pm$ needs to be chosen depending on the side of $x^*$ where the initial state $x_0$ is and $t(x_0, \epsilon)$ the time.

Then the \emph{Regularized Reaching Time} is
\begin{align}
\rrt{} =& \lim\limits_{\epsilon\rightarrow 0} t(x_0, \epsilon) - t(x_{ref}, \epsilon) \\
=& \int_{x_0}^{x_{ref}}\frac{dx}{f(x)}. \label{app1-eq:trr-1d}
\end{align}

In the same manner, the closed expression for \audic{} in one dimension can be derived:
\begin{align}
D(x) = \int_{x_0}^{x^*}dx\frac{d(x,x^*)}{f(x)} \label{app1-eq:audic-1d}
\end{align}
where $d(\cdot,\cdot)$ is the chosen distance function.

\subsection{Quadratic Correction}
\label{app1-sec:quadratic}

\begin{subequations}
	
	In order to demonstrate how the derived equations can be applied, we analyze this system which has a linear term plus a quadratic correction:
	\begin{align}
	&\dot{x} = f(x) = - x + b  x^2 \label{app1-eq:quadratic-system}
	\end{align}
	The attractor in this system is a fixed point at $x^* = 0$ and the corresponding basin of attraction is
	$\mathcal{B}(0) = \left( -\infty, 1/b \right)$.
	
	\rrt{} can be calculated in a straightforward manner using \Cref{app1-eq:trr-1d} and yields in \Cref{app1-eq:rrt-quadratic}. This result is depicted in \Cref{app1-fig:rrt-D-quadratic} (with different reference states $x_{ref}$ and compared with the linear result) and one can observe a different behavior one each side of $x^*$. Particularly relevant are the negative divergence at $x^* = 0$ and the positive one at $x_b =  1/b$. The latter fixed point is the boundary of the basin of attraction and hence never reaches $x^*$. So we expect \rrt{} in the limit to $x_b$ to diverge and this can be seen in \Cref{app1-fig:rrt-D-quadratic}.
	
	Furthermore, the curve intersects the x-axis at $x_{ref}$ as expected.
	\begin{align}
	T_{RR}(x) =& \log\left|\frac{x}{x_{ref}}\right| - \log\left|\frac{1-bx}{1-bx_{ref}}\right| \label{app1-eq:rrt-quadratic}
	\end{align}
\end{subequations}			
\begin{figure}
	\centering
	\includegraphics[width = 0.95\columnwidth]{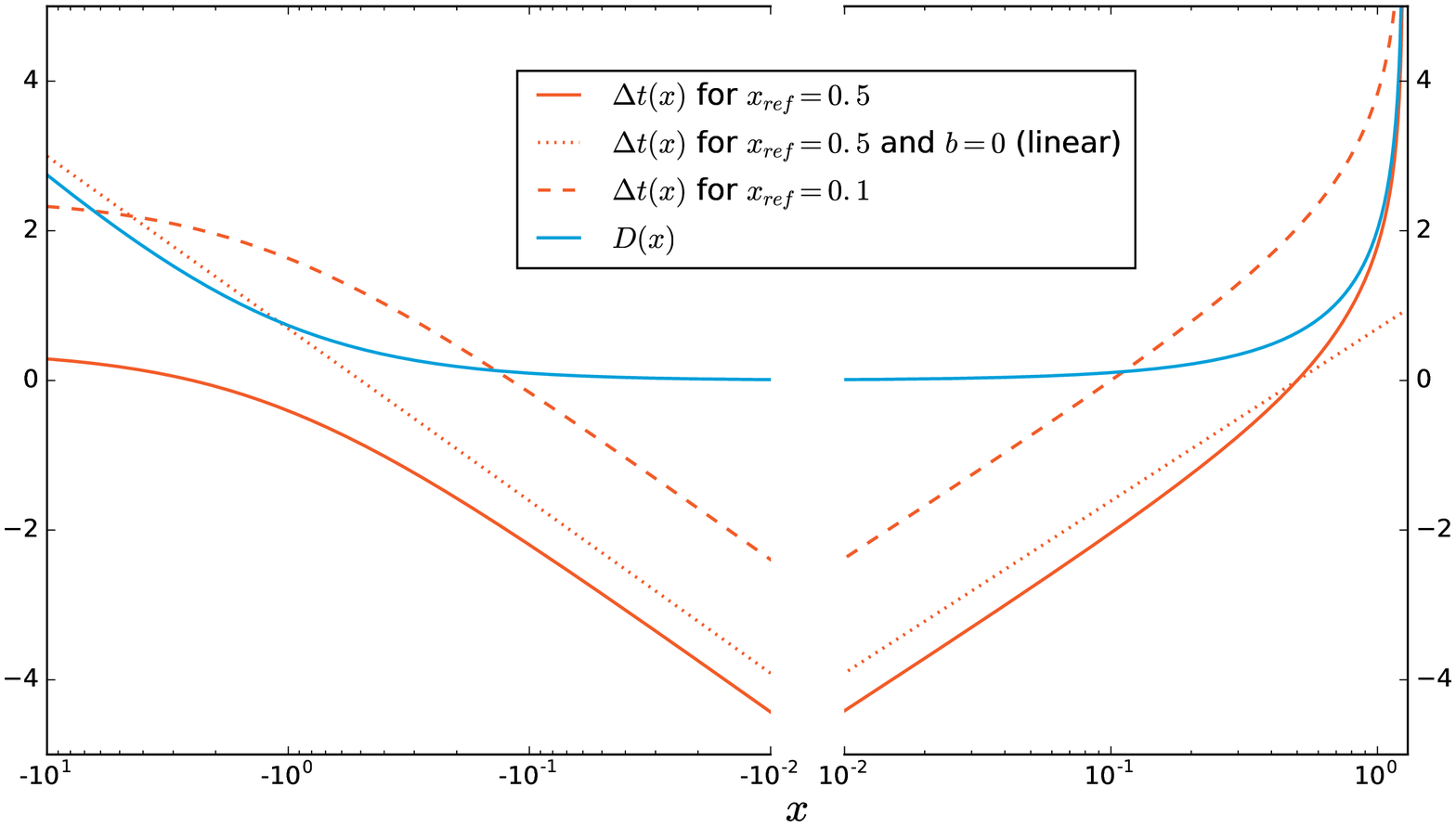}
	\caption{Analytical result for the \emph{Regularized Reaching Time} (\rrt{}) and \emph{Area under Distance Curve} (\audic{}, $D$) in a one-dimensional ODE with quadratic Right-Hand-Side (\Cref{app1-eq:quadratic-system},  $b = \frac{4}{5}$). Both metrics diverge as expected at the Basin boundary $x = \frac{1}{b}$. Furthermore, different reference points have been chosen to show that this shifts the resulting function only. In the symmetric log plot, the \rrt{} result for the linear system is a straight line and we can see how the quadratic term has an influence in comparison.}
	\label{app1-fig:rrt-D-quadratic}
\end{figure}	

\audic{} can be computed with the closed expression from \Cref{app1-eq:audic-1d}, too:
\begin{align}
D(x_0) &= - \frac{1}{b} \ln|1-bx_0|\quad x_0 \in \left[0, b\inv\right). \label{app1-eq:audic-1d-quadratic}
\end{align}
This result is depicted in \Cref{app1-fig:rrt-D-quadratic}. Particularly the divergence at $x \longrightarrow \frac{1}{b}$ is visible and due to reaching the basin boundary, analogously to \rrt{}.

At the attractor $x^* = 0$ \audic{} reaches $0$ because it never deviates from there, thus the cumulative distance vanishes.

\section{Rössler attractor}
\label{app1-sec:roessler-attractor}

\begin{figure*}
	\centering
	\includegraphics[width = 0.95\textwidth]{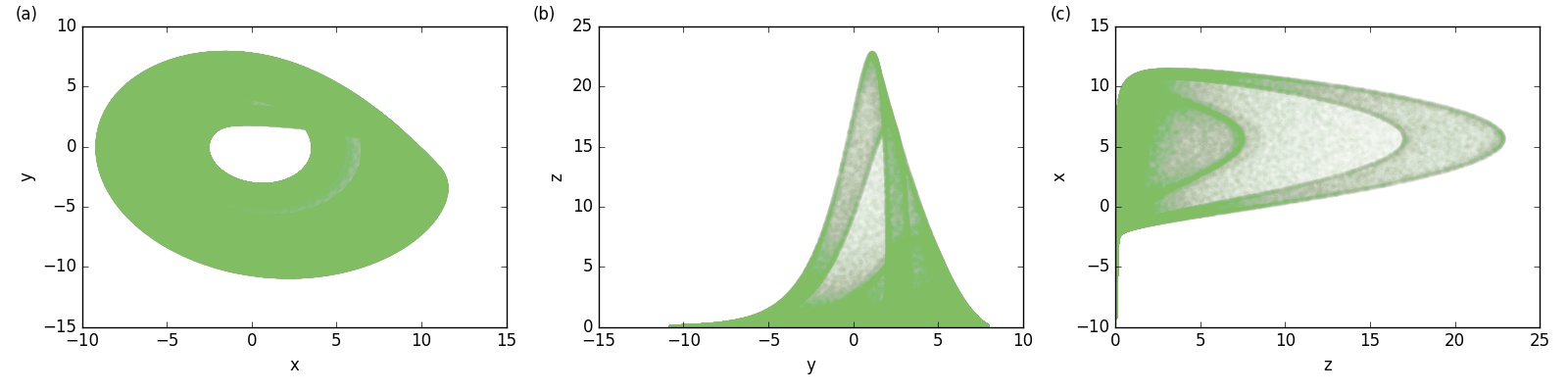}
	\caption{The three projections of the chaotic Rössler attractor for comparison with the phase space plots of \rrt{} and \audic{} in \Cref{fig:roessler-rrt,fig:roessler-audic}.}
	\label{app1-fig:roessler-attractor-projections}
\end{figure*}

\begin{figure}
	\centering
	\includegraphics[width = 0.95\columnwidth]{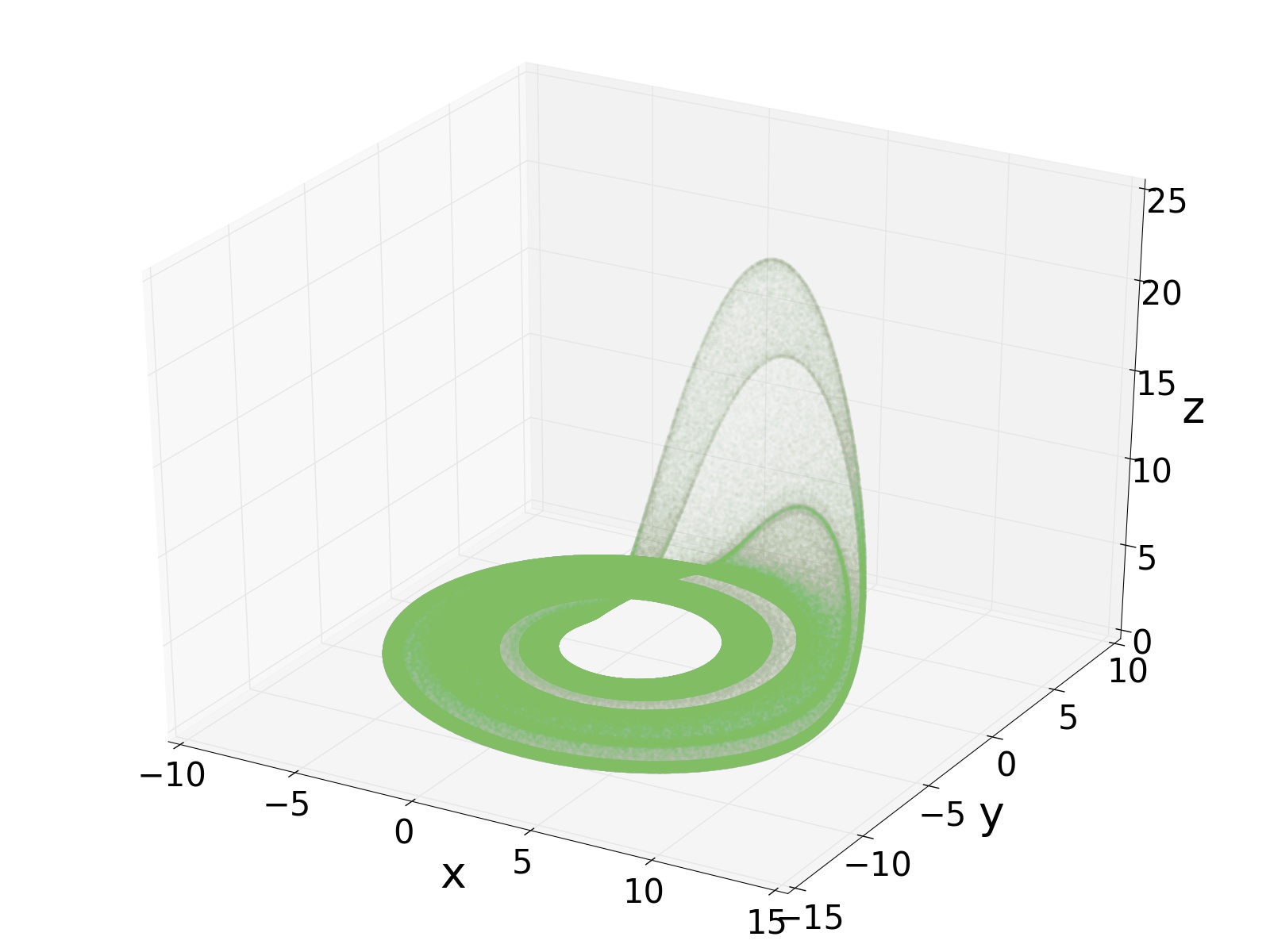}
	\caption{3D-plot of the chaotic Rössler attractor for comparison with the phase space plots of \rrt{} and \audic{} in \Cref{fig:roessler-rrt,fig:roessler-audic}.}
	\label{app1-fig:roessler-attractor-3d}
\end{figure}

For comparison with \Cref{fig:roessler-rrt,fig:roessler-audic}, the projections of the attractor has been plotted in \Cref{app1-fig:roessler-attractor-projections,app1-fig:roessler-attractor-3d}.

\section{Short Algorithm Description for Regularized Reaching Time}
\label{app1-sec:algo}

For complex systems like the Rössler attractor, a calculation of the \emph{Regularized Reaching Time} can get rather tricky.

The first problem is the estimation of the attractor itself. Finally, this could be solved starting at various different points and numerically integrate for a very long time. Removing the transient and then sampling the trajectories lead to a good estimate that could still fit in the available memory.

The second problem was the distance estimation. Fortunately, KD-Trees are exactly made for this and implemented in Scientific Python \cite{scipy}.

With these ingredients, \audic{} could be calculated. In order to estimate \rrt{}, the cumulative distances for the points along the trajectory were calculated backwards. This gives us the corresponding levelset of \audic{} for each point on the trajectory. With this, the times for entering different \audic{} levelsets could be retrieved and compared with the reference trajectory. Thus several values for \rrt{} were obtained and a limit could be estimated.

\section{Convergence of AuDiC}
\label{app1-sec:conv-audic}

The convergence of \audic{} depends on two elements: the distance function $d$ and the asymptotic approaching behavior of the trajectories. A common case with a mathematical distance function, i.e. a function fulfilling \Cref{app1-eq:dist-1,app1-eq:dist-2,app1-eq:dist-3,app1-eq:dist-4} \cite{heitzig2002mappings} and an exponentially stable attractor, the convergence can be proven right away as the bigger integral over the exponential envelope converges. 
\begin{subequations}
\begin{align}
\text{non-negativity} \quad &  d(x, y) \geq 0 \label{app1-eq:dist-1}\\
\text{identity of indiscernibles} \quad &  d(x, y) = 0 \Longleftrightarrow x = y \label{app1-eq:dist-2}\\
\text{symmetry} \quad & d(x, y) = d(y, x) \label{app1-eq:dist-3}\\
\text{triangle inequality} \quad & d(x, z) \leq d(x, y) + d(y, z)\label{app1-eq:dist-4}
\end{align}
\end{subequations}
Note the usage of the word mathematical distance (fulfilling the four properties). As we use the distance function only to measure how far a point is away from the attractor, more general functions could be used as well as long as they converge to $0$ when a trajectory approaches the attractor. Particularly, cost or damage functions that could be well motivated from the system's context are unlikely to always fulfill the requirements of a mathematical distance function.

Even assuming a mathematical distance function, convergence is not necessarily given. Systems that converge slower than exponentially could lead to a divergence in \audic{}. A simple example of such a case is $ \dot{x} = f(x) = \frac{-x^3}{2}$. The solutions are $\pm \frac{1}{\sqrt{c + t}}$ where the constant $c$ is fixed by the initial condition. Using the absolute values as the distance function gives a divergence for \audic{}, as $\int_0^\infty dt\ \frac{1}{\sqrt{c + t}} \longrightarrow \infty$.

\section{Convergence of \rrt}
\label{app1-sec:conv-rrt}

While the values of \rrt{} characterize the transient behavior of the trajectory, the existence of the limit in the definition \Cref{app1-eq:def-rrt} (\Cref{app1-eq:def-rrt}) depends actually on the asymptotic behavior.
\begin{align}
\rrt (x_0) := & \lim_{\epsilon \rightarrow 0} \left(t^\text{\audic}\left(x_0, \epsilon\right) - t^\text{\audic}\left(x_{ref}, \epsilon\right) \right). \label{app1-eq:def-rrt}
\end{align}

The simplest case is systems with finite reaching times because the RHS of \Cref{eq:def-rrt} can be split in two limits that converge separately. The results is the difference of the actual reaching times and is expected from the approach. Still, as the reaching times are finite anyway the complex approach with \rrt{} is not necessary and is just to show that it is reasonable in these cases, too.

For infinite reaching times, the values for $t^\text{\audic}$ will become increasingly large in the limit $\epsilon \rightarrow 0$. So the limit in \Cref{app1-eq:def-rrt} exists only, if for small $\epsilon$ (i.e. large times) the changes in $t^\text{\audic}\left(x_0, \epsilon\right)$ and $t^\text{\audic}\left(x_{ref}, \epsilon\right)$ will be about the same. This means, that for two different, small $\epsilon_1 > \epsilon_2$
\begin{align}
t^\text{\audic}\left(x_0, \epsilon_2\right) - &t^\text{\audic}\left(x_0, \epsilon_1\right) \approx \\
& t^\text{\audic}\left(x_{ref}, \epsilon_2\right) - t^\text{\audic}\left(x_{ref}, \epsilon_1\right) \nonumber
\end{align} 
Turning this interpretation around, it means that the trajectories have to behave ``similarly'' in the asymptotic limit, i.e. close to the attractor. 

The simplest case is a system with a hyperbolic fixed point where the larges eigenvalue of the corresponding Jacobian is real and of multiplicity 1. If that is the case, the calculation in \ref{app1-sec:linear-calculation} can be used locally around the attractor to understand why it converges and the precise proofs are in \ref{app2-sec:RRT}. Having a multiplicity larger than one might be mathematically interesting but is physically rather unlikely because some slight differences in the modeling of the system would usually change these. If this is a persistent property of the system, a precise understanding of the meaning is needed. Note that \rrt{} still converges but will depend on the underlying distance function. The latter is equivalent to problem (IV), non-invariance, and hence the result for these systems should be interpreted with care.

In the case of the largest eigenvalue being complex, the convergence can still be proven but there is a need for the choice of a specific distance function, as shown in \ref{app2-sec:RRT}. This is not problematic as the result is also invariant under change of variables and hence, simply the unique result that can be taken.

On the other hand, this suggest that our current definition might have to be improved as for more complex and higher-dimensional attractors we can currently rely only on numerics. The results for the Rössler attractor, particularly the strong sensitivity to the dynamics of the system as shown in \Cref{fig:roessler-bifurc} provide the numerical support for our current approach.

An improved version of the definition could be done using the properties as Lyapunov functions with constant negative orbital derivate. Even though we can give a rough outline of how to do that in the Discussion of the main paper, there are many subtle technicalities to be addressed in order to define an improved \rrt{} precisely.

%% file: Kittel-TimingOfTransients-appendix2.tex
\section{Precise Definitions and Theorems}
\label{app2-sec:definitions-and-theorems}
	
	We consider a deterministic dynamical system of the form
	\begin{equation}		\label{app2-eq:dynsys}
	\dot{x} = f(x),\qquad x\in\mathbb{R}^d,\quad f\in C^2(\mathbb{R}^d,\mathbb{R}^d)
	\end{equation}
	and assume that the system contains an exponentially asymptotically stable equilibrium $\mathcal{A}$. That is, $f(\mathcal{A}) = 0$ and all the eigenvalues of $Df(\mathcal{A})$ have negative real part.  We denote the basin of attraction of $\mathcal{A}$ by $\mathcal{B}_\mathcal{A}$. Also we denote the flow operator of \eqref{app2-eq:dynsys} as $\varphi(t,\cdot)$.
	
	Let the spectrum  of $Df(\mathcal{A})$ be $\lambda^s\cup \sigma^{ss}$, where $\lambda^s$ may be real or complex, but we assume has multiplicity one.  It will also be useful later to define constants $\alpha^s$ and $\alpha^{ss}$ such that
	\begin{equation*}
	\{\sigma^{ss}\} < \alpha^{ss} < \text{Re } \lambda^s < \alpha^s < 0,
	\end{equation*}
	where we also require $2|\alpha^s| > |\lambda^s|$.


	\subsection{Area under distance curve (AuDiC) function}	\label{app2-sec:AUDIC}
	
	Let $d(\cdot,\cdot)$ be a  metric defined on $\mathcal{B}_\mathcal{A}$. The AuDiC function is defined as
	\begin{equation}		\label{app2-eq:AuDiC}
	D(x_0) = \int_0^\infty d(x(t),\mathcal{A}) dt,
	\end{equation}
	where $x_0\in \mathcal{B}_\mathcal{A}$ and $x(t)$ is the solution to \eqref{app2-eq:dynsys} with $x(0)=x_0$. Under mild conditions on the metric $d$, the AuDiC function is a Lyapunov function.
	\begin{definition}			\label{app2-def:classK}
		A continuous function $\alpha: [0,\infty) \rightarrow [0,\infty)$ is a class $\mathcal{K}$ function if $\alpha(0) = 0$ and $\alpha$ is strictly monotonically increasing.
	\end{definition}
	Then the following result holds. We refer to \cite[Theorem 2.46]{giesl2007construction} for a proof.
	\begin{proposition}	\label{app2-prop:AUDICLyap}
		Let $d(\cdot,\mathcal{A})$ be a $C^1$ function and suppose that there is a class $\mathcal{K}$ function $\alpha$ such that $d(x,\mathcal{A}) \ge \alpha(||x-\mathcal{A}||_2)$ for all $x\in \mathcal{B}_\mathcal{A}$.
		
		Then $\frac{d}{dt} D(x) = -d(x,\mathcal{A})$ for $x\in \mathcal{B}_\mathcal{A}$. That is, the AuDiC function is a Lyapunov function with orbital derivative equal to $-d(x,\mathcal{A})$.
	\end{proposition}


	\subsection{Regularized return time (RRT) function}			\label{app2-sec:RRT}
	
	In the following, we further assume that the function $d(\cdot, \mathcal{A})$ from equation \eqref{app2-eq:AuDiC} has been chosen such that $d(x, \mathcal{A}) = || x - \mathcal{A}||_\mathcal{N}$ for some norm $|| \cdot ||_\mathcal{N}$. Then it is clear from \eqref{app2-eq:AuDiC} that the AuDiC function defines a norm $|| \cdot ||_D$ given by $D(x_0) = || x_0 - \mathcal{A} ||_D$.
	
	For a given initial condition $x_r\in \mathcal{B}_\mathcal{A}$, we denote the time taken for an initial condition to enter and remain inside a $D$-ball of radius $\epsilon$ as
	\begin{equation*}
	t^\text{\audic}(x_0,\epsilon):=\inf\{T: ||\varphi(t,x_0) - \mathcal{A}||_{D} <\epsilon \textrm{ for all }t\ge T\}.
	\end{equation*}
	The \rrt{} function is then defined as follows.
	\begin{definition}[\rrt{} function]		\label{app2-def:RRT}
		For a given reference point $x_r\in \mathcal{B}_\mathcal{A}$, the \rrt{} function is defined as
		\begin{equation}	\label{app2-eq:RRT}
		\rrt{}(x_0; x_r) := \lim_{\epsilon\rightarrow0}\left[t^\text{\audic}(x_0,\epsilon) - t^\text{\audic}(x_r,\epsilon)\right],
		\end{equation}
		where the limit exists.
	\end{definition}
	
	
	A natural question is under what conditions  the limit in \eqref{app2-eq:RRT} exists. To answer this question we distinguish between  two cases according to whether $\lambda^s$ is real or complex.
	Our first result is regarding the existence of the \rrt{} function in both cases, and the dependence on the choice of norm $||\cdot||_{D}$. In order to state the result for $\lambda^s$ complex, we make the following definition.
	
	\begin{definition}
		We define the following equivalence class on  norms defined on $\mathcal{B}_\mathcal{A}$:
		\begin{equation}	\label{app2-eq:equivnorm}
		||\cdot||_{\mathcal{N}_1} \sim ||\cdot||_{\mathcal{N}_2} \quad \Leftrightarrow \quad ||v||_{\mathcal{N}_1} = ||v||_{\mathcal{N}_2} \textrm{ for all }v\in E^s,
		\end{equation}
		where $E^s$ is the invariant subspace corresponding to the leading eigenvalue of $Df(\mathcal{A})$. 
	\end{definition}
	
	Clearly, elements in the above equivalence class are defined by the norm of elements in $E^s$. 
	
	\begin{proposition}		\label{app2-prop:RRTnorm}
		Let the \rrt{} function be defined as in Definition \ref{app2-def:RRT} for the system \eqref{app2-eq:dynsys}, and assume $x_0\not\in W^{ss}(\mathcal{A})$. Then we have the following
		\begin{enumerate}
			\item When $\lambda^s$ is real, the limit \eqref{app2-eq:RRT} exists for all choices of norm $||\cdot||_{\mathcal{N}}$. Moreover, the limit is independent of the choice of norm.
			\item When $\lambda^s$ is complex, the limit \eqref{app2-eq:RRT} exists if and only if $||\cdot||_{D} \sim ||\cdot||_P$, where
			$|| x ||_P := ||P^{-1} x||_2$,  $||\cdot||_2$ is the Euclidean 2-norm, and $P^{-1}Df(\mathcal{A})P$ is the Jordan normal form of $Df(\mathcal{A})$. 
		\end{enumerate}
	\end{proposition}
	
	We will also show that the \rrt{} function is closely related to the strong stable foliation $\mathcal{F}^{ss}$ in the basin of attraction of the equilibrium $\mathcal{A}$. We first recall the following definitions.
	
	\begin{definition}
		A  foliation $\mathcal{F}$ of an $d$-dimensional manifold $M$ is a partition of $M$ into a disjoint collection of $k$-dimensional injectively immersed connected submanifolds (called leaves) such that for each $x\in M$, there is a neighborhood $V\subset M$ and a chart
		\begin{equation*}
		\phi: V \rightarrow \mathbb{R}^k \times \mathbb{R}^{d-k},
		\end{equation*}
		such that each connected component of the intersection of a leaf of $\mathcal{F}$ with $V$ is mapped to the set $\mathbb{R}^k\times \{y\}$, for some  $y\in\mathbb{R}^{d-k}$.
		
		We call $\mathcal{F}$ a $C^r$ foliation if each local chart is $C^r$. A continuous foliation whose leaves are $C^r$ is called a $C^r$ lamination.
	\end{definition}
	
	We denote the leaf of a foliation through a point $x$ as $\mathcal{F}(x)$. A foliation $\mathcal{F}$ is invariant under the flow of \eqref{app2-eq:dynsys} if $\varphi(t,\mathcal{F}(x)) = \mathcal{F}({\varphi(t,x)})$ for sufficiently small $|t|$. 
	
	\begin{theorem}[\cite{hirsch1977invariant}]			\label{app2-thm:Fss}
		Consider the system \eqref{app2-eq:dynsys} and let $\mathbb{R}^d = E^s \oplus E^{ss}$ be the direct sum decomposition into the invariant stable and strong stable subspaces for the linear system $\dot{x} = Df(\mathcal{A})x$. Then there exists a unique invariant $C^r$ lamination $\mathcal{F}_{ss}$ in $\mathcal{B}_\mathcal{A}$, called strong stable foliation, such that each leaf of $\mathcal{F}_{ss}$ has dimension equal to $\text{dim } E^s$ and $\mathcal{F}_{ss}(\mathcal{A}) = W^{ss}(\mathcal{A})$, where $W^{ss}(\mathcal{A})$ is the strong stable manifold of $\mathcal{A}$.
		
		Solutions $x(t)$ and $y(t)$ that belong to the same leaf of $\mathcal{F}_{ss}$ for all time are characterized by strong asymptotic convergence to each other: $||x(t) - y(t)||_D \le Ce^{-\alpha^{ss}t}$ for $t$ sufficiently large.
	\end{theorem}
	
	We will also prove the following result which provides an important characterization of the \rrt{} function.
	
	\begin{proposition}			\label{app2-prop:RRTfoliation}
		Let the \rrt{} function be defined as in Definition \ref{app2-def:RRT} for the system \eqref{app2-eq:dynsys}, and assume $x_r\not\in W^{ss}(\mathcal{A})$. In the case $\lambda^s$ is complex, we assume $||\cdot||_{D}\sim||\cdot||_P$ as in  Proposition \ref{app2-prop:RRTnorm}. Then the level sets of $\rrt{}(x_0;x_r)$ are equal to the leaves of $\mathcal{F}_{ss}$. That is,
		\begin{equation*}
		\rrt{}(x_0;x_r) = \rrt{}(y_0;x_r) \Leftrightarrow \mathcal{F}_{ss}(x_0) = \mathcal{F}_{ss}(y_0)
		\end{equation*}
		Furthermore, $\rrt{}(x_0;x_r)\rightarrow -\infty$ as $x_0$ approaches $\mathcal{F}_{ss}(\mathcal{A}) \,(=W^{ss}(\mathcal{A}))$.
	\end{proposition}
	
	The proof of Propositions \ref{app2-prop:RRTnorm} and \ref{app2-prop:RRTfoliation} rely on the following result regarding the behavior of solutions in the approach to equilibrium. We refer to \cite{sandstede1993verzweigungtheorie} for a proof.
	
	\begin{theorem}			\label{app2-thm:solutionrep}
		Consider the system \eqref{app2-eq:dynsys}, and define $\lambda^s$, $\alpha^s$ and $\alpha^{ss}$ as before. Then there exists $\kappa >0$ such that for all solutions $x(t)$ of \eqref{app2-eq:dynsys} in $\mathcal{B}_\mathcal{A}$ with $||x(0) - \mathcal{A}||_D < \kappa$, the limit
		\begin{equation}
		\eta(x(0)) := \lim_{t\rightarrow\infty} \Phi(0,t)P^s (x(t) - \mathcal{A})		\label{app2-eq:etadef}
		\end{equation}
		exists, where $\Phi(t,0)$ is the transition matrix of $\dot{x} = Df(\mathcal{A})x$ from $0$ to $t$ and $P^s$ is the projection onto $E^s$ along $E^{ss}$. Furthermore, we have the representation
		\begin{equation}
		x(t)-\mathcal{A} = \Phi(t,0)\eta(x(0)) + \mathcal{O}(e^{-\min\{|\alpha^{ss}|,2|\alpha^s|\}t}).	\label{app2-eq:solutionrep}
		\end{equation}
	\end{theorem}
	
	Note that since $\Phi(0,t)$ leaves $E^s$ invariant and $E^s$ is closed, we have $\eta(x(0))\in E^s$. It also follows from the proof in \cite{sandstede1993verzweigungtheorie} that $\eta:B_{\kappa}(\mathcal{A})\rightarrow E^s$ is continuous.
	
	\begin{lemma}			\label{app2-lem:Fsseta}
		Let $x(t),y(t)$ be solutions to \eqref{app2-eq:dynsys} in $\mathcal{B}_\mathcal{A}$.
		Then $x(t),y(t)$ belong to the same leaf of $\mathcal{F}_{ss}$ for all $t$ if and only if $\eta(x(s)) = \eta(y(s))$ for $s$ sufficiently large. Furthermore, $\eta(x(s)) = 0$ if and only if $x(t)\in W^{ss}(\mathcal{A})$.
	\end{lemma}
	
	\begin{proof}
		By Theorem \ref{app2-thm:Fss}, the solutions  $x(t)$ and $y(t)$ belong to the same leaf of $\mathcal{F}_{ss}$ for all $t\ge 0$ if and only if $||x(t) - y(t)||_D \le Ce^{-\alpha^{ss}t}$. Let $s>0$ be large enough so that $||x(s) - \mathcal{A}||_D,||y(s) - \mathcal{A}||_D<\kappa$. Now from Theorem \ref{app2-thm:solutionrep} we have $\eta(x(s)),\eta(y(s))\in E^s$, and equation \eqref{app2-eq:solutionrep} implies that this is possible if and only if $\eta(x(s)) = \eta(y(s))$. The last statement follows directly from \eqref{app2-eq:solutionrep} and the theory of stable/unstable manifolds.
	\end{proof}

	\vspace{8mm}
	
	\paragraph*{Proof of Proposition \ref{app2-prop:RRTnorm}.} Let $x_0$ and $x_r$ be as in Definition \ref{app2-def:RRT} and let $x(t),\hat{x}(t)$ be the solutions to \eqref{app2-eq:dynsys} with $x(0)=x_0$ and $\hat{x}(0) = x_r$. Let $s>0$ be large enough so that $||x(s) - \mathcal{A}||_D, ||\hat{x}(s) - \mathcal{A}||_D<\kappa$, and let $\epsilon>0$ be small. Then from Theorem \ref{app2-thm:solutionrep}, the solution $x(t)$ intersects the boundary of the $D$-ball $B_\epsilon(\mathcal{A})$ when
	\begin{equation}				\label{app2-eq:epsiloneq}
	|| \Phi(t,s)\eta(x(s)) + \mathcal{O}(e^{-\min\{|\alpha^{ss}|,2|\alpha^s|\}t}) ||_{D}   =\epsilon 	
	\end{equation}
	
	
	We first consider the case where $\lambda^s$ is real. Then $\eta(x(s))\in E^s$ is one dimensional and we obtain
	\begin{eqnarray}			
	t - \frac{1}{\lambda^s}\ln\epsilon &= &  -  \frac{1}{\lambda^s}\ln ||\eta(x(s)) + g(t)||_{D} =:F_{re}(t),\label{app2-eq:epsiloneqreal1}\\
	t - \frac{1}{\lambda^s}\ln\epsilon &= &  -  \frac{1}{\lambda^s}\ln ||\eta(\hat{x}(s)) + \hat{g}(t)||_{D} =: \hat{F}_{re}(t),\label{app2-eq:epsiloneqreal2}
	\end{eqnarray}
	where $g(t),\hat{g}(t) = \mathcal{O}(e^{-\delta t})$ for some $\delta >0$. 
	For $\epsilon>0$ sufficiently small,  equations \eqref{app2-eq:epsiloneqreal1} and \eqref{app2-eq:epsiloneqreal2} can be uniquely solved for ${t}(\epsilon)$ and $\hat{t}(\epsilon)$ respectively, and  ${t}(\epsilon),\hat{t}(\epsilon)\rightarrow\infty$ as $\epsilon\rightarrow 0$. We refer to Figure \ref{app2-fig:real} for a sketch of the solutions to equations \eqref{app2-eq:epsiloneqreal1} and \eqref{app2-eq:epsiloneqreal2}.
	Then $t^\text{\audic}_{}(x_0,\epsilon) = t(\epsilon)$ and ${t}^\text{\audic}(x_r, \epsilon) = \hat{t}(\epsilon)$ and we have
	{\small
		\begin{eqnarray}
	t^\text{\audic}_{}(x_0,\epsilon) - t^\text{\audic}_{}(x_r,\epsilon) &=&  \frac{1}{\lambda^s} \ln {\textstyle\frac{||\eta(\hat{x}(s)) - \hat{g}(\hat{t}(\epsilon))||_{D}}{||\eta({x}(s)) - {g}({t}(\epsilon))||_{D}}}	\nonumber \\
	\lim_{\epsilon\rightarrow 0}\left[t^\text{\audic}_{}(x_0,\epsilon) - t^\text{\audic}_{}(x_r,\epsilon)\right] &= & \frac{1}{\lambda^s} \ln{\textstyle \frac{||\eta(\hat{x}(s))||_{D}}{||\eta({x}(s))||_{D}}}.		\label{app2-eq:RRTreal}
	\end{eqnarray}
	}
	Recall $\eta(\hat{x}(s)),\eta(x(s))\in E^s$ is one dimensional, so then we have
	\begin{equation*}
	\rrt{}(x_0;x_r) = \frac{C(x_0)}{\lambda^s}
	\end{equation*}
	for some $C(x_0)\in\mathbb{R}$, which is independent of the norm $||\cdot||_{D}$ (and hence $||\cdot||_\mathcal{N}$). This proves the first part of Proposition \ref{app2-prop:RRTnorm}.
	
	\begin{figure*}	
		\subfloat[Schematic diagram illustrating solutions of \eqref{app2-eq:epsiloneqreal1} and \eqref{app2-eq:epsiloneqreal2}.]{
			\psfrag{X}{$\Gamma$}
			\begin{overpic}[width=0.9\columnwidth,natwidth=876,natheight=658]{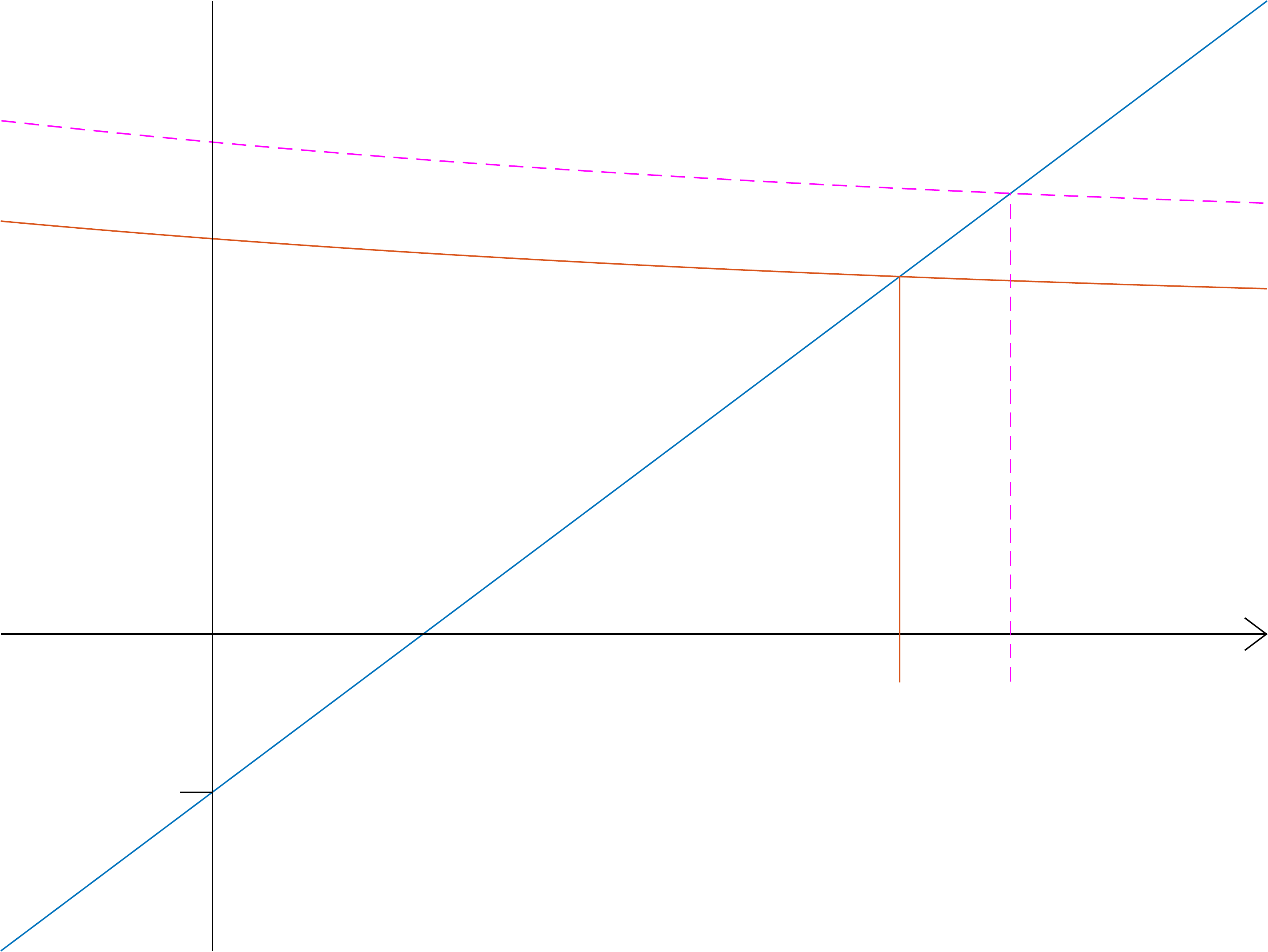}
				\put(2.7,26){$\mathcal{B}_\mathcal{A}$}
				\put(0,12){\footnotesize{$ - \frac{1}{\lambda^s}\ln \epsilon$}}
				\put(13,27){\small{$0$}}
				\put(100,21){\small{$t$}}
				\put(68,17){\small{${t}(\epsilon)$}}
				\put(77,17){\small{$\hat{t}(\epsilon)$}}
				\put(93,48){\footnotesize{$F_{re}(t)$}}
				\put(93,61){\footnotesize{$\hat{F}_{re}(t)$}}
			\end{overpic}
			\label{app2-fig:real}}	
		\hfill	
		\subfloat[Schematic diagram illustrating solutions of \eqref{app2-eq:complexteq1} and \eqref{app2-eq:complexteq2} for a choice of $||\cdot||_{\mathcal{N}}$ not in the equivalence class of $||\cdot||_P$, so that $t(\epsilon)-\hat{t}(\epsilon)$ oscillates as $\epsilon\rightarrow0$ (see text).]{	
			\psfrag{X}{$\Gamma$}
			\begin{overpic}[width=0.9\columnwidth,natwidth=691,natheight=590]{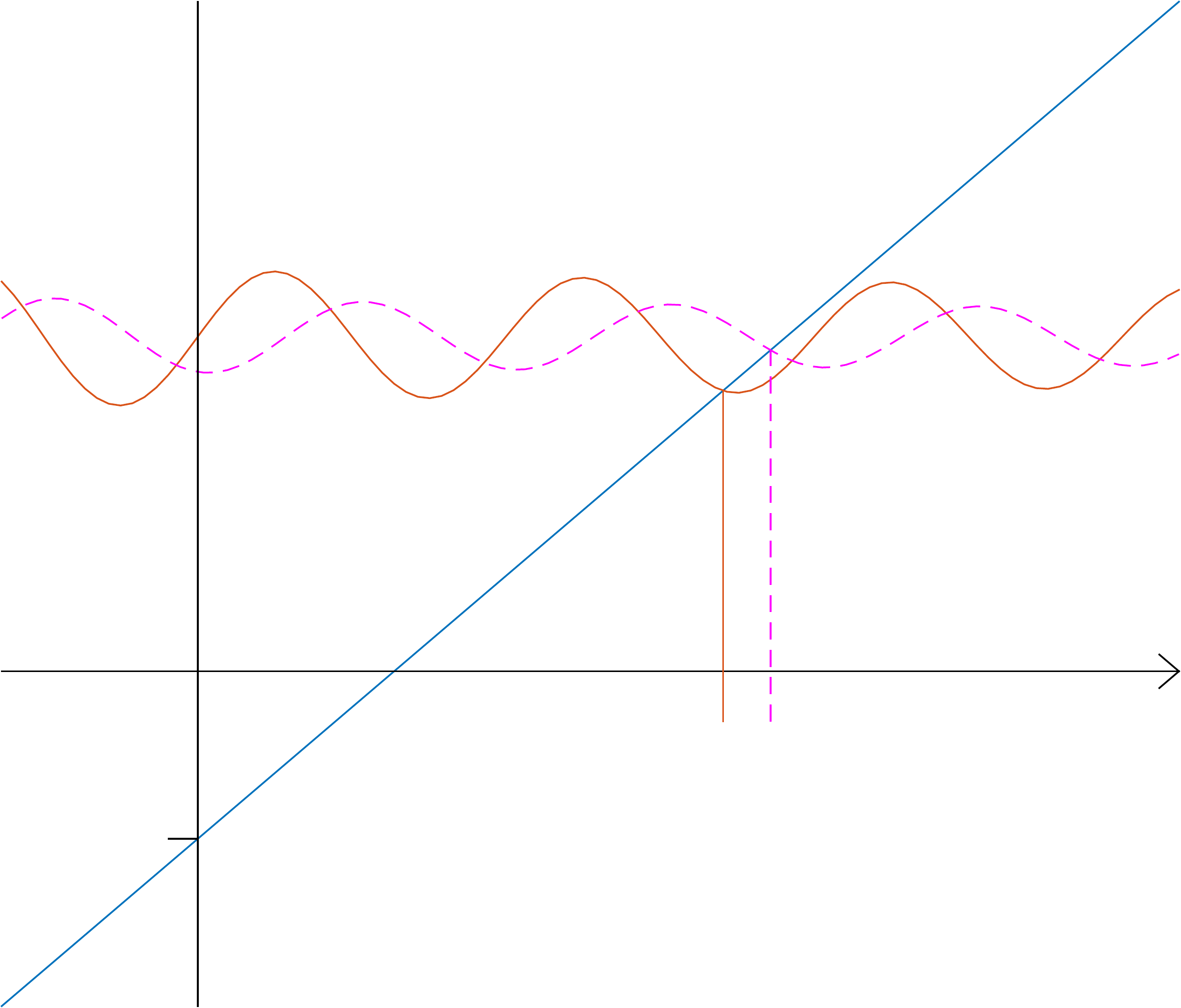}
				\put(-5,13.5){\footnotesize{$ - \frac{1}{\text{Re }\lambda^s}\ln \epsilon$}}
				\put(13,30){\small{$0$}}
				\put(100,21){\small{$t$}}
				\put(57,20){\small{${t}(\epsilon)$}}
				\put(65,20){\small{$\hat{t}(\epsilon)$}}
				\put(96,62){\footnotesize{$F_{com}(t)$}}
				\put(96,50){\footnotesize{$\hat{F}_{com}(t)$}}
			\end{overpic}
			\hskip 0.08 \columnwidth
			\label{app2-fig:complex}}
		\caption{Schematic Diagrams illustrating the different kinds of solutions.}
	\end{figure*}

	Now assume that $\lambda^s$ is complex. Then equation \eqref{app2-eq:epsiloneq} leads to
	{\small
	\begin{align}	 		
	&\begin{aligned}
	t - \frac{1}{\text{Re }\lambda^s}\ln\epsilon=&   - \frac{1}{\text{Re }\lambda^s}\ln \left|\left|P\Psi(t,s)P^{-1}\eta(x(s)) + g(t)\right|\right|_{D}\\
	=&: F_{com}(t),
	\end{aligned}	\label{app2-eq:complexteq1}\\
	&\begin{aligned}
	t - \frac{1}{\text{Re }\lambda^s}\ln\epsilon =&   - \frac{1}{\text{Re }\lambda^s}\ln \left|\left|P\Psi(t,s)P^{-1}\eta(\hat{x}(s)) + \hat{g}(t)\right|\right|_{D}\\
	 =&: \hat{F}_{com}(t),
	\end{aligned}	\label{app2-eq:complexteq2}
	\end{align}
}
	where 
	$P^{-1} Df(\mathcal{A})P$ is the Jordan normal form of $Df(\mathcal{A})$, and $e^{\mu t}\Psi(t,s)$ is the transition matrix of $\dot{w}=P^{-1} Df(\mathcal{A})P w$ from $s$ to $t$.  Again $g(t),\hat{g}(t) = \mathcal{O}(e^{-\delta t})$ for some $\delta>0$. However now the term $P\Psi(t,s)P^{-1}\eta(x(s))$ oscillates as $t\rightarrow\infty$, although it will remain bounded. Then for $\epsilon>0$ small enough, the above equation will have a (not necessarily unique) solution for $t$. We define $t(\epsilon)$ as the largest such solution for given $\epsilon$, then again $t(\epsilon)\rightarrow\infty$ as $\epsilon\rightarrow 0$ and $t^\text{\audic}_{}(x_0,\epsilon) = t(\epsilon)$. We refer to Figure \ref{app2-fig:complex} for a sketch of the solutions to equations \eqref{app2-eq:complexteq1} and \eqref{app2-eq:complexteq2}.
	Then we have 
	\begin{equation*}
	\begin{aligned}
	&t^\text{\audic}_{}(x_0,\epsilon) - t^\text{\audic}_{}(x_r,\epsilon) \\
	&\quad = \frac{1}{\text{Re }\lambda^s}\ln\frac{ \left|\left|P\Psi(\hat{t}(\epsilon),s)
		P^{-1}\eta(\hat{x}(s)) + \hat{g}(\hat{t}(\epsilon))\right|\right|_{D}}{ \left|\left|P\Psi(t(\epsilon),s)
		P^{-1}\eta(x(s)) + g(t(\epsilon))\right|\right|_{D}}
	\end{aligned}
	\end{equation*}

	Note that $||\Psi(t,s) P^{-1}\eta(x(s))||_2 = ||P^{-1}\eta(x(s))||_2$ since $\eta(x(s))\in E^s$. Now if $||\cdot||_{D} \sim ||\cdot||_P$ as in Proposition \ref{app2-prop:RRTnorm}, then  we have
	\begin{equation}
	\begin{aligned}
	&\lim_{\epsilon\rightarrow 0}\left[t^\text{\audic}_{}(x_0,\epsilon) - t^\text{\audic}_{}(x_r,\epsilon)\right] \\
	& \qquad \qquad \qquad = \frac{1}{\text{Re }\lambda^s}\ln\frac{||P^{-1} \eta(\hat{x}(s))||_2}{||P^{-1} \eta({x}(s))||_2}
	\end{aligned}		\label{app2-eq:RRTcomplex}
	\end{equation}
	so the \rrt{} function is well defined. However, for any other choice of norm $||\cdot||_{D}$, the quantity $||P\Psi({t}(\epsilon),s)
	P^{-1}\eta({x}(s))||_{D}$ will not converge to a constant value, and  will oscillate as $t(\epsilon)\rightarrow\infty$. In particular, for any $\eta(x(s))$ with $\eta(x(s))\ne C \eta(\hat{x}(s))$, the function $t^\text{\audic}_{}(x_0,\epsilon) - t^\text{\audic}_{}(x_r,\epsilon)$ will not converge as $\epsilon\rightarrow 0$. This proves the second part of Proposition \ref{app2-prop:RRTnorm}.\hfill\qed
	
	\paragraph*{Proof of Proposition \ref{app2-prop:RRTfoliation}.} From Lemma \ref{app2-lem:Fsseta}, we have that 
	\begin{equation*}
	x_0,y_0 \in \mathcal{F}_{ss}(x_0)  \quad\Leftrightarrow\quad \eta(x(s)) = \eta(y(s))
	\end{equation*}
	for $s$ sufficiently large. From \eqref{app2-eq:RRTreal} and \eqref{app2-eq:RRTcomplex} we have
	\begin{equation*}
	\rrt{}(x_0;x_r) - \rrt{}(y_0;x_r) =  \frac{1}{\text{Re }\lambda^s} \ln \frac{||\eta(y(s))||_{D}}{||\eta({x}(s))||_{D}}.		
	\end{equation*}
	and so it follows that
	\begin{equation*}
	\rrt{}(x_0;x_r) = \rrt{}(y_0;x_r)\quad\Leftrightarrow\quad  x_0,y_0 \in \mathcal{F}_{ss}(x_0). 
	\end{equation*}
	The final statement of Proposition \ref{app2-prop:RRTfoliation} follows from $\eta(x(s))=0 \Leftrightarrow x(t)\in W^{ss}(\mathcal{A})$ (see Lemma \ref{app2-lem:Fsseta}) and the fact that $\eta:B_\kappa(\mathcal{A})\rightarrow E^s$ is continuous. \hfill\qed

%% file: Kittel-TimingOfTransients.bbl
\providecommand{\newblock}{}
\begin{thebibliography}{10}
\expandafter\ifx\csname url\endcsname\relax
  \def\url#1{{\tt #1}}\fi
\expandafter\ifx\csname urlprefix\endcsname\relax\def\urlprefix{URL }\fi
\providecommand{\eprint}[2][]{\url{#2}}

\bibitem{lenton2011early}
Lenton T~M 2011 {\em Nature Climate Change\/} {\bf 1} 201--209

\bibitem{scheffer2009early}
Scheffer M, Bascompte J, Brock W~A, Brovkin V, Carpenter S~R, Dakos V, Held H,
  Van~Nes E~H, Rietkerk M and Sugihara G 2009 {\em Nature\/} {\bf 461} 53--59

\bibitem{anderies2013topology}
Anderies J~M, Carpenter S~R, Steffen W and Rockstr{\"o}m J 2013 {\em
  Environmental Research Letters\/} {\bf 8} 044048

\bibitem{yuan2003solution}
Yuan Y, Kubokawa J and Sasaki H 2003 {\em Power Systems, IEEE Transactions
  on\/} {\bf 18} 1094--1102

\bibitem{rossler1976equation}
R{\"o}ssler O~E 1976 {\em Physics Letters A\/} {\bf 57} 397--398

\bibitem{giesl2007construction}
Giesl P 2007 {\em Construction of global Lyapunov functions using radial basis
  functions\/} (Springer)

\bibitem{fiutak1980transient}
Fiutak J and Mizerski J 1980 {\em Zeitschrift f{\"u}r Physik B Condensed
  Matter\/} {\bf 39} 347--352

\bibitem{tang1975transient}
Tang C, Telle J and Ghizoni C 1975 {\em Applied Physics Letters\/} {\bf 26}
  534--537

\bibitem{barkema1994transient}
Barkema G, Marko J and De~Boer J 1994 {\em EPL (Europhysics Letters)\/} {\bf
  26} 653

\bibitem{castellano2009statistical}
Castellano C, Fortunato S and Loreto V 2009 {\em Reviews of modern physics\/}
  {\bf 81} 591

\bibitem{chowdhury2000statistical}
Chowdhury D, Santen L and Schadschneider A 2000 {\em Physics Reports\/} {\bf
  329} 199--329

\bibitem{krapivsky2010kinetic}
Krapivsky P~L, Redner S and Ben-Naim E 2010 {\em A kinetic view of statistical
  physics\/} (Cambridge University Press)

\bibitem{van2007long}
Van~Geest G, Coops H, Scheffer M and van Nes E 2007 {\em Ecosystems\/} {\bf 10}
  37--47

\bibitem{hastings2004transients}
Hastings A 2004 {\em Trends in Ecology \& Evolution\/} {\bf 19} 39--45

\bibitem{schaffer1993transient}
Schaffer W~M, Kendall B, Tidd C~W and Olsen L~F 1993 {\em Mathematical Medicine
  and Biology\/} {\bf 10} 227--247

\bibitem{fisher1989disequilibrium}
Fisher F~M 1989 {\em Disequilibrium foundations of equilibrium economics\/} 6
  (Cambridge University Press)

\bibitem{fischer2005epileptic}
Fisher R~S, Boas W~v~E, Blume W, Elger C, Genton P, Lee P and Engel J 2005 {\em
  Epilepsia\/} {\bf 46} 470--472 ISSN 1528-1167

\bibitem{van2016constrained}
van Kan A, Jegminat J, Donges J~F and Kurths J 2016 {\em Phys. Rev. E\/} {\bf
  93}(4) 042205

\bibitem{kuznetsov2013elements}
Kuznetsov Y~A 2013 {\em Elements of applied bifurcation theory\/} vol 112
  (Springer Science \& Business Media)

\bibitem{scheffer2012anticipating}
Scheffer M, Carpenter S~R, Lenton T~M, Bascompte J, Brock W, Dakos V, Van
  De~Koppel J, Van De~Leemput I~A, Levin S~A, Van~Nes E~H {\em et~al.\/} 2012
  {\em Science\/} {\bf 338} 344--348

\bibitem{menck2013basin}
Menck P~J, Heitzig J, Marwan N and Kurths J 2013 {\em Nature Physics\/} {\bf 9}
  89--92

\bibitem{menck2014dead}
Menck P~J, Heitzig J, Kurths J and {Joachim Schellnhuber} H 2014 {\em Nature
  Communications\/} {\bf 5:3969} 1--8 ISSN 2041-1723

\bibitem{klinshov2015stability}
Klinshov V~V, Nekorkin V~I and Kurths J 2015 {\em New Journal of Physics\/}
  {\bf 18} 013004

\bibitem{hellmann2015survivability}
Hellmann F, Schultz P, Grabow C, Heitzig J and Kurths J 2016 {\em Scientific
  reports\/} {\bf 6:29654} 1--12

\bibitem{mitra2015integrative}
Mitra C, Kurths J and Donner R~V 2015 {\em Scientific reports\/} {\bf 5:16196}
  1--10

\bibitem{nolting2011grundkurs}
Nolting W 2011 {\em Grundkurs {Theoretische} {Physik} 3\/} Springer-{Lehrbuch}
  (Berlin, Heidelberg: Springer Berlin Heidelberg) ISBN 978-3-642-13448-7
  978-3-642-13449-4

\bibitem{cvitanovic2016chaos}
Cvitanovi{\'c} P, Artuso R, Mainieri R, Tanner G and Vattay G 2016 {\em Chaos:
  Classical and Quantum\/} (Copenhagen: Niels Bohr Inst.)
  \urlprefix\url{http://ChaosBook.org/}

\bibitem{mauroy2013isostables}
Mauroy A, Mezi{\'c} I and Moehlis J 2013 {\em Physica D: Nonlinear Phenomena\/}
  {\bf 261} 19--30

\bibitem{josic2006isochron}
Josic K, Shea-Brown E~T and Moehlis J 2006 {\em Scholarpedia\/} {\bf 1} 1361

\bibitem{heck2016esd}
Heck V {\em et~al.\/} 2016 {\em Earth System Dynamics (in prep.)\/}

\bibitem{schultz2014detours}
Schultz P, Heitzig J and Kurths J 2014 {\em New Journal of Physics\/} {\bf 16}
  125001

\bibitem{zgliczynski1997computer}
Zgliczynski P 1997 {\em Nonlinearity\/} {\bf 10} 243

\bibitem{barrio2011qualitative}
Barrio R, Blesa F, Dena A and Serrano S 2011 {\em Computers \& Mathematics with
  Applications\/} {\bf 62} 4140--4150

\bibitem{rockstrom2009planetary}
Rockstr{\"{o}}m J, Steffen W~L, Noone K, Persson b, {Chapin III} F~S, Lambin E,
  Lenton T~M, Scheffer M, Folke C, Schellnhuber H~J and Others 2009 {\em
  Ecology and Society\/} {\bf 14}

\bibitem{steffen2015planetary}
Steffen W, Richardson K, Rockstr{\"o}m J, Cornell S~E, Fetzer I, Bennett E~M,
  Biggs R, Carpenter S~R, de~Vries W, de~Wit C~A {\em et~al.\/} 2015 {\em
  Science\/} {\bf 347} 1259855

\bibitem{havlin2012challenges}
Havlin S, Kenett D~Y, Ben-Jacob E, Bunde A, Cohen R, Hermann H, Kantelhardt J,
  Kert{\'e}sz J, Kirkpatrick S, Kurths J {\em et~al.\/} 2012 {\em The European
  Physical Journal Special Topics\/} {\bf 214} 273--293

\bibitem{giesl2015review}
Giesl P and Hafstein S 2015 {\em Discrete and Continuous Dynamical Systems,
  Series B\/} {\bf 20}(8) 2291--2331

\bibitem{python}
Van~Rossum G and Drake~Jr F~L 1995 {\em Python reference manual\/} (Centrum
  voor Wiskunde en Informatica Amsterdam)

\bibitem{numpy}
Ascher D, Dubois P~F, Hinsen K, Hugunin J, Oliphant T {\em et~al.\/} 2001
  Numerical python

\bibitem{scipy}
Jones E, Oliphant T, Peterson P {\em et~al.\/} 2001 {SciPy}: Open source
  scientific tools for {Python} [Online; accessed 2016-05-10]
  \urlprefix\url{http://www.scipy.org/}

\bibitem{heitzig2002mappings}
Heitzig J 2002 {\em Mappings Between Distance Sets Or Spaces\/} Ph.D. thesis
  Universit{\"a}t Hannover

\bibitem{hirsch1977invariant}
Hirsch M~W, Shub M and Pugh C~C 1977 {\em Invariant manifolds\/} ({\em Lecture
  Notes in Mathematics\/} vol 587) (Springer)

\bibitem{sandstede1993verzweigungtheorie}
Sandstede B 1993 {\em Verzweigungtheorie homokliner Verdopplungen\/} Ph.D.
  thesis University of Stuttgart

\end{thebibliography}
